\documentclass[final]{siamltex}
\usepackage{amsmath, amssymb}

\title{A Stochastic Compartmental Model for Fast Axonal Transport}

\author{Lea Popovic\thanks{Department of
Mathematics and Statistics, Concordia University, Montreal QC H3G 1M8} \and Scott A.~McKinley\thanks{Department of Mathematics, University of Florida, 358 Little Hall Box 118105, Gainesville, FL 32605 ({\tt
scott.mckinley@ufl.edu})} \and Michael
C.~Reed\thanks{Department of Mathematics, Duke University, Box
90320, Durham, NC 27701}}

\begin{document}

\maketitle

\begin{abstract}
In this paper we develop a probabilistic micro-scale compartmental model and use it to study macro-scale properties of axonal transport, the process by which intracellular cargo is moved in the axons of neurons. By directly modeling the smallest scale interactions, we can use recent microscopic experimental observations to infer all the parameters of the model.  Then, using techniques from probability theory, we compute asymptotic limits of the stochastic behavior of individual motor-cargo complexes, while also characterizing both equilibrium and non-equilibrium ensemble behavior. We use these results in order to investigate three important biological questions: (1) How homogeneous are axons at stochastic equilibrium? (2) How quickly can axons return to stochastic equilibrium after large local perturbations?  (3) How is our understanding of delivery time to a depleted target region changed by taking the whole cell point-of-view? 

\end{abstract}

\pagestyle{myheadings} \thispagestyle{plain}

\newcounter{other}
\newtheorem{remark}[other]{Remark}

\setlength{\unitlength}{1 cm}

\def\ep{\epsilon}
\def\gen{\mathcal{A}}
\def\ev{\mathbf{e}}
\def\xv{\mathbf{x}}
\def\pv{\mathbf{p}}
\def\qv{\mathbf{q}}
\def\Pv{\mathbf{P}}
\def\Qv{\mathbf{Q}}
\def\ddt{\frac{d}{dt}}
\def\ep{\varepsilon}
\newcommand{\E}[1]{\mathbb{E}{\left[ #1\right]}}
\newcommand{\V}[1]{\mathrm{Var}{\left[ #1\right]}}
\newcommand{\p}[1]{\mathbb{P}{\left\{ #1\right\}}}
\def\gbf{\mathbf{g}}
\def\lbf{\boldsymbol{\lambda}}
\def\lambdabf{\boldsymbol{\lambda}}
\def\ebf{\mathbf{e}}

\section{Introduction}

In all cells, one finds that proteins, membrane-bound organelles,
and other structures (e.g. chromosomes) are transported from place
to place at speeds much higher than diffusion. Though these
transport processes are fundamental to cell function, many of the
underlying mechanisms, organizational principles, and regulatory
features remain unknown. Axonal transport is one of the best studied
systems because the transport is basically one-dimensional since
axons are long and narrow. There are two speeds of axonal transport.
Fast transport goes at speeds of roughly $0.2$ to $0.5$ meters/day
\cite{lasek82}\cite{ochs72}, while slow transport goes at
approximately 1 millimeter/day, the rate of axon growth and
regeneration \cite{brown00}\cite{lasek82}. The biology and
principles of slow transport are not yet clear \cite{brown00}, but
the basic mechanisms of fast axonal transport were discovered in the
1980s \cite{allen85}\cite{blum85}\cite{miller85}\cite{vale85}. The
model in this paper refers to fast axonal transport, which we will
henceforth call axonal transport.

The axonal transport apparatus consists of vesicles which form
reversible chemical bonds with motor proteins that bind reversibly
to microtubules which run parallel to the long dimension of the axon
\cite{alberts08}. When the vesicle-motor protein complex is
assembled on the microtubule, the complex steps stochastically with
step size approximately 8 nanometers for kinesin and dynein and 10
nanometers for myosin
\cite{carter05}\cite{finer94}\cite{gennerich07}\cite{svoboda93}. The
vesicles enter from the cell body on microtubules and then detach
and reattach to the transport mechanism at random times.

In this paper we propose a spatial Markov-chain compartmental model based
on these dynamics.  We will assume independence of the interactions,
and exponential wait times between events.  While we address the
validity of these assumptions in the Discussion section, we
consider this a useful ``first-order'' approximation that permits
study of the dynamics from both the perspective of individual
vesicles as well as that of the full spatial system. Such a model
unifies all earlier deterministic and stochastic modeling efforts and can accommodate both
qualitative and quantitative experimental data observed on multiple
scales.

In much experimental work in the 1970s and 1980s, radio-labeled
amino acids were put into the cell bodies continuously or for a few
hours. The amino acids were incorporated into proteins that were
packaged into vesicles and put on the transport system so that at
later times radioactivity could be seen moving progressively down
the axons. In the continuous infusion case, one would see a wave of
radioactivity with a sharp but slowly spreading wavefront
propagating at constant velocity down the axon. In the case of
infusion for a few hours one would see at long times a slowly
spreading pulse of radioactivity that looked normally distributed.
It was to understand this behavior that Reed and Blum constructed
PDE models for axonal transport
\cite{blum85}\cite{reed86}\cite{reed94}. These models did not have
traveling wave solutions, but the data certainly looked like
approximate traveling waves. In \cite{reed90} it was shown by a
perturbation theory argument that, in the asymptotic limit where the
unbinding and binding rates $k_2$ and $k_1$ get large, the solution
approaches a slowly spreading traveling wave or a normal pulse.
Recently, in a series of papers, Friedman and co-workers have
introduced new PDE models and proved these results
rigorously\cite{friedman05}\cite{friedman06}\cite{friedman07}\cite{friedman07-2}.

Probabilistic models for axonal transport were introduced and used
for simulations already in the 1980s
\cite{stewart82}\cite{takenaka84}. However, rigorous work began with
Lawler \cite{lawler95} in 1995 and was continued by Brooks
\cite{B99} who used a continuous time stochastic model to show that
the distribution of an individual particle is a spreading Gaussian
at large times. Brooks also proved tail estimates for the central
limit theorem and used them to estimate the error from normal.
Independently, Bressloff \cite{bressloff06} developed a discrete stepping model and performed an analysis under the assumption that the rate of unbinding and binding to transport is fast relative to lateral velocity over the length scale of interest.  The author derived a characterization of the spreading wavefront of a particle entering at the nucleus and traveling to the distal end.  This model served as the basis for later investigations by Newby and Bressloff \cite{newby09, newby10} wherein the authors characterize the axonal transport system as an intermittent search for hidden targets.  

In this paper we revise the existing probabilistic models in order
to study randomness in the system as a whole rather than
exclusively from the point of view of an individual particle.  Our goal is
not only to recover and generalize previous results, but also to investigate three specific, biologically important, properties of the whole stochastic system. 

\subsection{Summary of Results}
\label{subsec:summary-results} 

In Section \ref{subsec:model}, we create a continuous-time Markov chain queueing model for the axonal transport system.  We show how to use experimental data to determine (or estimate) all the parameters of the model.

In Section \ref{sec:particle-perspective}, we take the individual vesicle point-of-view. We prove the asymptotic forms in \cite{reed90} with rigorous error estimates. We show that in the limit as the compartment size becomes small our model becomes the probabilistic model of \cite{B99}. We also show that in the limit as the length of the axon and time become large (with the scale of axon length on the order of the squared scale of time) our model becomes the PDE model of \cite{reed94}. 
Since we assume that particles are independent, the time evolution of the  law of an individual will reflect the behavior of an ensemble of particles released at the same time.

In Section \ref{sec:system-persective} we adopt the
full spatial system perspective to quantify stochasticity along the
length of the axon. We begin by calculating in Proposition
\ref{prop:equilib} the stationary distribution of a flow-through
system that has sustained input from the nucleus, while particles
are removed upon reaching the distal end.  The stationary
distribution has a product Poisson structure which allows for
seamless transition between spatial scales.  

With this mathematical model, we are able to make precise statements about three biologically important system properties.  In the stationary distribution, the number of vesicles in each slice of the axon is independent and identically distributed (Proposition \ref{prop:equilib}), however this in and of itself is not sufficient to account for the sense that samples taken for different parts of the cell ``look the same.''  We compute in Section \ref{subsec:homog-at-equilib} the coefficient of variation for the number of vesicles in sections of different length and show  that the coefficient of variation is low for all but the smallest length scales. In Section \ref{sec:target}, we study the intermittent search problem posed by Newby and Bressloff \cite{newby09} from the system point of view.  Efficient transport to locations that need material must balance the speed of transport of material from the nucleus to the distal end of the cell with the rates of dissociation from the transport apparatus.  We calculate the expected hitting time for a hidden target by all vesicles in the system. In so doing we encounter the counterintuitive result that while increasing the velocity of the motors while on transport increases the chance of any particular vesicle missing the target, the expected hitting time by the system actually decreases.

This hitting time approach is natural for needed material that is sparsely distributed throughout the axon, but when the needed cargo in question is more common, the time to replenishment is better addressed through the ODE approach that we develop in Section \ref{sec:approach-to-equilib}.  Due to the product structure of the law of the transient dynamics, this non-equilibrium behavior is determined by the $2N$-dimensional ODE governing the means. From this we estimate the timescale of return to equilibrium as a function of the length scale of interest.

\

\section{The model and its parameters}
\label{subsec:model}

Let $L$ be the length of the axon, divided evenly into $N=L/\delta$ lateral
sections each of length $\delta$, equal to the step size of the
motor protein. Within each section, we disregard any further spatial
geometry and take the particles to be in one of two states:
\begin{itemize}
\item[$\cdot$] an on-transport state that steps laterally at a rate $r=v/\delta$ per section, or
\item[$\cdot$] an off-transport state that does not step laterally.
\end{itemize}


\begin{figure}[h]
\begin{center}
\begin{picture}(11, 5)

    \linethickness{0.075 mm}

    \thicklines
    \put(1,1){\line(1,0){10}}
    \put(1,4){\line(1,0){10}}

    \thinlines
    \multiput(4.5,1.15)(3.75,0){2}{
    \multiput(-0.05,0)(0,0.5){6}{\line(0,1){0.2}}}

    \put(0.6,0.5){$x = 0$}
    \put(3.95,0.5){$x = \frac{L}{3}$}
    \put(7.7,0.5){$x = \frac{2L}{3}$}
    \put(10.5,0.5){$x = L$}

    \put(2.45,4.7){$Q_1 = 3$}
    \put(2.5,4.2){$P_1 = 3$}
    \put(5.85,4.7){$Q_2 = 1$}
    \put(5.9,4.2){$P_2 = 2$}
    \put(9.45,4.7){$Q_3 = 3$}
    \put(9.5,4.2){$P_3 = 2$}

    \put(0,1.25){\line(0,1){2.5}}
    \put(0.7,1.25){\line(0,1){2.5}}
    \put(0,1.25){\line(1,0){0.7}}
    \put(0,3.75){\line(1,0){0.7}}

    \put(0.21,3.3){N}
    \put(0.24,3.0){u}
    \put(0.26,2.7){c}
    \put(0.28,2.3){l}
    \put(0.26,2.0){e}
    \put(0.23,1.7){u}
    \put(0.27,1.4){s}


    \put(0.7, 3.65){\line(1,0){1.1}}
    \put(0.7, 3.2){\line(1,0){0.6}}
    \put(0.7, 2.3){\line(1,0){0.8}}
    \put(0.7, 1.7){\line(1,0){1.4}}
    \put(0.7, 1.3){\line(1,0){1}}

    \multiput(2,2)(4,-0.7){2}{\line(1,0){1.5}}
    \multiput(3.7,3.7)(3,-0.25){2}{\line(1,0){2}}
    \multiput(1.2,2.8)(1.4,-0.5){4}{\line(1,0){0.7}}
    \multiput(1.6,1.2)(1.4,0.5){4}{\line(1,0){0.6}}
    \multiput(2.0,1.4)(2,-0.3){2}{\line(1,0){1.5}}
    \multiput(2.8,2.9)(2,0.3){2}{\line(1,0){1.5}}
    \multiput(6.2,2.1)(1.4,0.5){3}{\line(1,0){0.7}}
    \multiput(6.9,2.35)(1,-0.5){3}{\line(1,0){1.2}}
    \multiput(7.3,3.2)(1.4,-0.5){3}{\line(1,0){0.9}}
    \put(1.4,3.3){\line(1,0){1.2}}
    \put(9.7,1.6){\line(1,0){1.2}}
    \put(9.3,3.6){\line(1,0){1.1}}

    \put(2,1.9){$\bullet$}
    \put(9.4,3.0){$\bullet$}
    \put(10.2,2.1){$\bullet$}
    \put(1,3.1){$\bullet$}
    \put(4.1,1.7){$\bullet$}
    \put(6.3,1.2){$\bullet$}
    \put(8.6,1.75){$\bullet$}
    \put(2.5,1.6){$\ast$}
    \put(3.2,3.6){$\ast$}
    \put(5.1,2.4){$\ast$}
    \put(6.8,2.7){$\ast$}
    \put(3.6,2.5){$\ast$}
    \put(9.3,2.2){$\ast$}
    \put(10.1,2.6){$\ast$}

\end{picture}


\begin{picture}(10.5,6)
    \linethickness{0.075 mm}

    \thicklines

    \put(-0.5,4.5){\vector(1,0){1}}
    \put(-0.3,4.65){$q_0 r$}

    \put(0.5,0.5){\line(0,1){2}}
    \put(0.5,0.5){\line(1,0){2}}
    \put(2.5,0.5){\line(0,1){2}}
    \put(0.5,2.5){\line(1,0){2}}
    \put(0.7,0.7){$P_1$}
    \put(0.6,2.85){$k_2$}
    \put(2.1,2.85){$k_1$}
    \put(0.7,5){$Q_1$}

    \qbezier(1.2,3.5)(0.9,3)(1.2,2.5)
    \qbezier(1.8,3.5)(2.1,3)(1.8,2.5)
    \put(1.1,2.7){\vector(1,-2){0.1}}
    \put(1.9,3.3){\vector(-1,2){0.1}}

    \put(0.5,3.5){\line(0,1){2}}
    \put(0.5,3.5){\line(1,0){2}}
    \put(2.5,3.5){\line(0,1){2}}
    \put(0.5,5.5){\line(1,0){2}}

    \put(2.5,4.5){\vector(1,0){2}}
    \put(3.45,4.65){$r$}

    \put(4.5,0.5){\line(0,1){2}}
    \put(4.5,0.5){\line(1,0){2}}
    \put(6.5,0.5){\line(0,1){2}}
    \put(4.5,2.5){\line(1,0){2}}

    \put(4.7,0.7){$P_2$}
    \put(4.6,2.85){$k_2$}
    \put(6.1,2.85){$k_1$}
    \put(4.7,5){$Q_2$}

    \qbezier(5.2,3.5)(4.9,3)(5.2,2.5)
    \qbezier(5.8,3.5)(6.1,3)(5.8,2.5)
    \put(5.1,2.7){\vector(1,-2){0.1}}
    \put(5.9,3.3){\vector(-1,2){0.1}}

    \put(4.5,3.5){\line(0,1){2}}
    \put(4.5,3.5){\line(1,0){2}}
    \put(6.5,3.5){\line(0,1){2}}
    \put(4.5,5.5){\line(1,0){2}}

    \put(6.5,4.5){\vector(1,0){2}}
    \put(7.45,4.65){$r$}

    \put(8.5,0.5){\line(0,1){2}}
    \put(8.5,0.5){\line(1,0){2}}
    \put(10.5,0.5){\line(0,1){2}}
    \put(8.5,2.5){\line(1,0){2}}

    \put(8.7,0.7){$P_3$}
    \put(8.6,2.85){$k_2$}
    \put(10.1,2.85){$k_1$}
    \put(8.7,5){$Q_3$}

    \qbezier(9.2,3.5)(8.9,3)(9.2,2.5)
    \qbezier(9.8,3.5)(10.1,3)(9.8,2.5)
    \put(9.1,2.7){\vector(1,-2){0.1}}
    \put(9.9,3.3){\vector(-1,2){0.1}}

    \put(8.5,3.5){\line(0,1){2}}
    \put(8.5,3.5){\line(1,0){2}}
    \put(10.5,3.5){\line(0,1){2}}
    \put(8.5,5.5){\line(1,0){2}}

    \put(10.5,4.5){\vector(1,0){1}}
    \put(10.9,4.65){$r$}

\end{picture}
\end{center}
\caption{On-and-off Transport Chain} \label{fig:double-chain}
\end{figure}
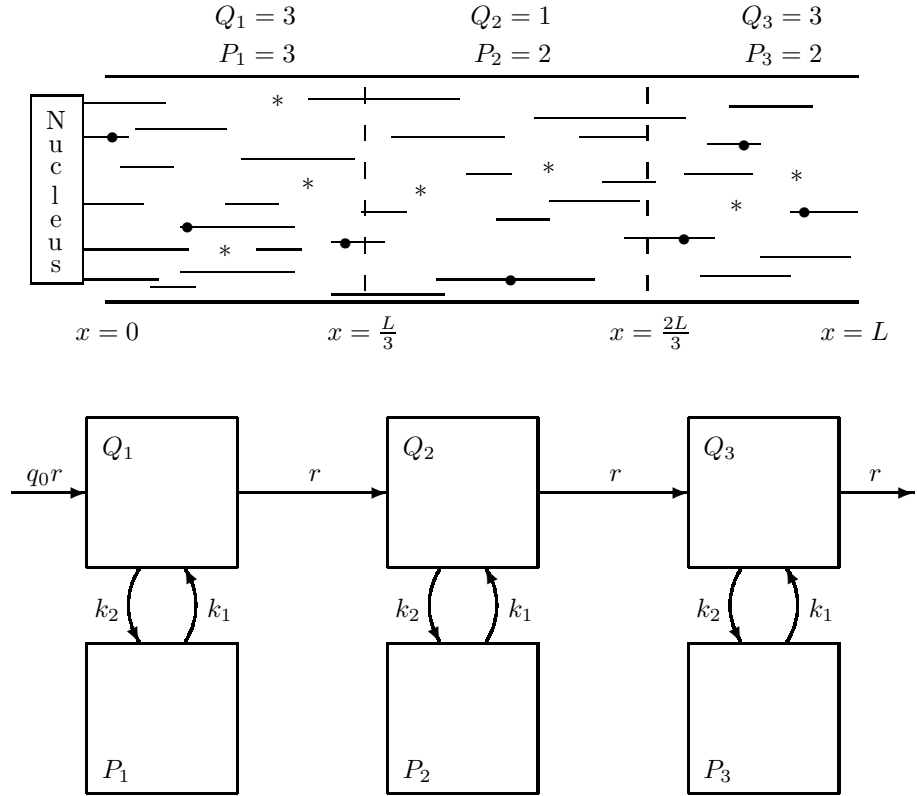

\noindent We use a $2N$-dimensional continuous time Markov chain $\{(Q_i(t),P_i(t)), i=1,\dots, N\}$ to model the
particle dynamics, where
\begin{itemize}
\item[$\cdot$] $Q_i(t)$ is the number of particles at time $t$ in on-transport state in section
$i$,
\item[$\cdot$] $P_i(t)$ is the number of particles at time $t$ in off-transport state in section $i$.
\end{itemize}
\begin{definition}\hspace{-1mm}\emph{(Stochastic compartmental model)}
\label{def-particle} Let $(\{(Q_i(t),P_i(t)), i=1,\dots, N\})_{t\ge 0}$ be a continuous time Markov chain on 
the state space $\mathbb{N}^{N}\!\times\mathbb{N}^{N}$ with the following transitions and time dependent rates:
\begin{itemize}
\item[$\cdot$] {(Lateral transport)} \\
$(Q_{i}, Q_{i+1}) \to (Q_i -1, Q_{i+1}+1)$ at rate $r Q_i(t)$;
\item[$\cdot$] {(Switch from on-transport to off-transport)} \\
$(Q_{i},P_i) \to (Q_i - 1, P_{i}+1)$  at rate $k_2 Q_i(t)$;
\item[$\cdot$] {(Switch from off-transport to on-transport)} \\
$(P_{i},Q_i) \to (P_i - 1, Q_{i} +1)$ at rate $k_1 P_i(t)$;
\item[$\cdot$] {(Production of new particles)}\\
$(Q_1) \to (Q_1 + 1)$ at rate $r q_0$;
\item[$\cdot$] {(Removal of particles at distal end)}\\
$(Q_N) \to (Q_N - 1)$ at rate $r Q_N(t)$.
\end{itemize}
\end{definition}

The lateral transport rate, $r = v/\delta$, is inversely proportional to the
length scale so that the mean number of particles per unit length is
invariant with respect to rescaling $\delta.$ We will assume that the rate of
production $q_0=\delta\rho_0$, for some constant $\rho_0>0$, in order for the mean number of particles in each compartment to scale with the size  $\delta$ of each compartment. This will imply that the mean number of particles per unit length scales as $\rho_0$.  A graph of the model is depicted in Figure 1.

In order to insure the Markov property, we use exponential random
variables for the waiting times between transition events.
Specifically, we mean that after a given event we assign a new
independent random variable to each of the $3N + 2$ possible next
events, exponentially distributed with the appropriate rate
parameter. The system of values updates according to the transition
associated with the minimum of these waiting times.  Then we create
a new set of exponential random variables and the process proceeds
as before.

The advantage of computing
explicit formulas for quantities that can be observed in experiments
is that the experimental data can then be used to determine the
parameter values in the model. For the characterization of the
approximate wavefront speed and spreading in Section
\ref{sec:particle-perspective} and the homogeneity calculations in
Section \ref{sec:system-persective} we need order-of-magnitude
estimates for the parameters. Actual parameter values will certainly
differ depending on the particular neural tissue and the particular
particles being transported. However, we can get order-of-magnitude
estimates from existing data.

First we recall that fast transport has been observed to travel at
speeds of 0.2 to 0.5 $m$ per day.  We can assume that the average
velocity of particles while physically bound to microtubules is
roughly 1 $m$ per day, or $v=10^{-6}$ $m/s$.  We assume that the compartment size scales as the length scale of the individual steps of the motor protein, so $\delta \sim 10^{-8} m$.  This implies that the rate
parameter should be $r = v / \delta = 100 s^{-1}$.

We now turn our attention to the on-off rates rates $k_2$ and $k_1$.
These can be determined from experimentally observed run lengths on
the transport system. Indeed, Dixit et al.~\cite{dixit08} show that
a typical run along microtubules for dinein and kinesin is on the
order of $10^{-6} m$. We can compare this with the theoretical run
length of the model to determine off-rate $k_2$. Within the model,
at each step on the transport mechanism the particle has a binary
decision to jump laterally along the transport with probability $r /
(k_2 + r)$, or to jump off with probability $k_2/ (r + k_2)$. The
number of jumps along the transport system before jumping off is
therefore geometrically distributed on the set $\{0, 1, \ldots\}$
with success probability $\frac{r}{r + k_2}$.  It follows that
average number of steps in the run is $\frac{r}{k_2}$ and therefore
the average run length is $\frac{r}{k_2} \times 10^{-8} \, m$.
Setting this equal to the average experimental run length of
$10^{-6}$ from \cite{dixit08}, we see that $\frac{r}{k_2} \sim
100$, implying that $k_2 \sim 1 s^{-1}$. As we will see in the
computation of the stationary distribution in Section
\ref{subsec:system-stationary-dist}, the ratio of the expected
number of particles on the track to those off the track is
$\frac{k_2}{k_1}$. Dixit \cite{dixit08} found that approximately
75\% of the particles were motile so this ratio is approximately
equal to 3. Since $k_2 \sim 1 s^{-1}$ we see that $k_1 \sim
\frac{1}{3}.$

It remains to estimate $q_0$. We will see in Proposition
\ref{prop:equilib} that the mean number of particles per compartment
is $(1 + \frac{k_2}{k_1}) q_0 = 4 q_0$. Of course, axons have a
large variety of diameters and larger axons will have more vesicles
per unit length so one expects a range of values for $q_0$. However,
examination of a large number of electron micrographs of axonal
cross-sections (see for example \cite{gross82}, Fig.~3;
\cite{jastrow-url}; \cite{moranurl}), which are typically 100 $nm$
thick enables one to estimate the number of vesicles per 100 $nm$
segment. This number is typically in the range of 10 to 100 which
implies that there are 1 to 10 vesicles per compartment in our model.
Therefore $q_0$ is in the range 0.25 to 2.5, for various axons.

We remark that we are ignoring some aspects of the physics and the
biology of axonal transport. We are not including diffusion of the
vesicles off the track. We are treating the microtubule track as
though it were a single continuous entity from one end of the axon
to the other, when it fact it consists of numerous, separated,
microtubule fragments. And, we are ignoring retrograde transport and
the details of the motor proteins. Nevertheless, this simple model
will enable us to investigate the homogeneity questions that are the
main goal of this paper.

\

\section{Dynamics from the Particle Perspective}
\label{sec:particle-perspective}

In this section we calculate properties of the stochastic dynamics
by using stochastic convergence theorems and stochastic averaging theorems from probability theory.  We first see that, in the $\delta
\rightarrow 0$ limit, the law of the location of a single particle corresponds to that of a particle with a piecewise linear Markov motion. We then show this law can be approximated by the Green's function of a linear partial differential equation. This enables us to obtain, as a special case,  the asymptotic behavior of the PDE models for axonal transport in an either rate limiting or perturbed setting.

\subsection{The active transport mode}
\label{subsec:particle-active-transport}

 We first consider the simple case where the
particle starts at $X^\delta_0 = 0$ and stays exclusively in active transport mode. Let
$X^\delta_t \in \{0,\delta, 2 \delta, \ldots, L\}$ be the lateral position of
a particle at time $t$ and let $n^\delta_t$ be the number of jumps made by
the particle as of time $t$. Observe, $X^\delta = \delta n^\delta$.

\begin{lemma} \label{prop:active-transport}
Let $k_2 = k_1 = 0$, and $r = v/\delta>0$, then  
 the position of the particle satisfies, for any time $t<\infty$
\[\sup\limits_{s\leq t}\big|X^\delta_s-vs\big|\;\; \mathop{\longrightarrow}\limits_{\delta\to 0}\;\; 0 \; \; \mbox{ a.s.} \;\;\] 
\mbox{and}, for  $B$  a standard Brownian motion
\[\frac{1}{\sqrt{\delta}}\Big(X^\delta_t-vt\Big)_{t\geq 0}\;\;
\mathop{\Longrightarrow}\limits_{\delta\to 0}\;\; \sqrt{v} \big(B_t\big)_{t\geq
0}\]
in distribution on the Skorokhod space of cadlag (right continuous left limited) functions.
\end{lemma}

\begin{proof}
Since $n^\delta$ is a Poisson
process with rate $r={v}/{\delta}$, defining $N_t:=n^\delta_{\delta  t}$ we get a Poisson
process $N$ with rate $v$, and we have $X^\delta_\cdot=\delta N_{\cdot/\delta}$. Our results then follow directly from
the functional law of large numbers (FLLN) and the functional central limit theorem (FCLT)
for the Poisson process $N$.
\end{proof}

\subsection{The on/off dynamics}
\label{subsec:particle-on-off}

We now consider a particle which undergoes transitions from
on-transport to off-transport state and back. Denote again by $X^\delta_t
\in \{0, \delta, 2 \delta, \ldots, L\}$ the lateral position of a
particle at time $t$ and let $n^\delta_t$ be the number of lateral
transition jumps made by the particle as of time $t$. Observe that
the particle will spend only a fraction of its time in active
transport and hence the lateral speed of the particle should be
slower than before.

A non-compartmental stochastic model for axonal transport, introduced by Brooks in \cite{B99}, is as follows. A particle can be in one of two states:
\begin{itemize}
\item[$\cdot$] an on-transport state with deterministic lateral velocity $v$, or
\item[$\cdot$] an off-transport state with lateral velocity 0.
\end{itemize}
We use a 2-dimensional Markov process to model the particle dynamics, where
\begin{itemize}
\item[$\cdot$] $X_t$ be the  lateral position of this particle at time $t$, 
\item[$\cdot$] $\xi_t$ be the indicator for whether it is on ($1$) or off ($0$) transport at time $t$. 
\end{itemize}
\begin{definition}\hspace{-1mm}\emph{(Stochastic non-compartmental model)}
Let $(X_t,\xi_t)_{t\ge 0}$ be a piecewise-linear Markov process with values in $(\mathbf{R}_+,\{0,1\})$ started at $(X_0,\xi_0)=(0,1)$ with the following dynamics:
\begin{itemize}
\item[$\cdot$]{(Switch from on-transport to off-transport)}\\
$(X_t,1) \to (X_t,0)$ at rate $k_1$
\item[$\cdot$]{(Switch from off-transport to on-transport)}\\
$(X_t,0) \to (X_t,1)$ rate $k_2$
\item[$\cdot$]{(Lateral travel)}\hspace{1mm}
$X_t=\int_0^t v\xi_sds$
\end{itemize}
\end{definition}
The path of $(X_t)_{t\ge 0}$ consists of alternating sequence of Exponential($k_2$) stretches of time where the lateral position increases linearly with speed $v$, and Exponential($k_1$) stretches of time where it remains constant.
\begin{proposition} \label{prop:on-off}
Let $k_2,k_1>0$, and $r = v/\delta>0$, then the position of the particle converges
\[(X^\delta_t)_{t\ge 0} \mathop{\Longrightarrow}\limits_{\delta\to0} (X_t)_{t\ge 0}\]
 in distribution on the Skorokhod space of cadlag functions.
\end{proposition}

\begin{proof}
If we let $\xi^\delta$ be the indicator of whether the particle in the compartmental model is on ($\xi^\delta=1$) or off ($\xi^\delta=0$) transport, then $(X^\delta_t,\xi^\delta_t)_{t\ge 0}$ is a strong Markov process. We will see that $\xi^\delta$ is a continuous-time Markov chain  on $\{0,1\}$ (independent of $\delta$), and conditionally on $\xi^\delta$ the transition law of $X$ is easily expressed. Likewise, for the non-compartmental model above, $(X_t, \xi_t)_{t\ge 0}$ is a strong Markov process, with $\xi$ the same continuous-time Markov chain on $\{0,1\}$ as $\xi^\delta$, and conditionally on $\xi$ the change in $X$ is easily given in terms of its linear speed and $\xi$.

We will start by showing that for any $t>0$, $(X^\delta_t,\xi^\delta_t)$ converges to $(X_t,\xi_t)$ in distribution as $\delta\to 0$. We then show that the finite dimensional distributions of $(X^\delta, \xi^\delta)$ converge to those of $(X,\xi)$. A tightness argument finally implies (Lemma 16.2 and Theorem 16.3 \cite{Kal}) that $(X^\delta, \xi^\delta)$ converges to $(X,\xi)$ in distribution on the Skorokhod space of cadlag processes. 

Suppose that initially the particle is on transport at $x_0$, so $(X^\delta_0,\xi^\delta_0)=(x_0,1)$. The first time $\tau_1=\inf\{t>0: \xi^\delta_t=0\}$ at which the particle steps off transport has Exponential($k_2$) distribution, irrespective of $\delta$. The first subsequent increment in time $\sigma_1=\inf\{t>0: \xi^\delta_{\tau_1+t}=0\}$ after which the particle steps back on transport has Exponential($k_1$) distribution, irrespective of $\delta$ as well. This is repeated, and $\xi^\delta$ is a simple continuous-time Markov chain  on $\{0,1\}$ with transition rates $k_2$ and $k_1$, from $1\to 0$ and $0 \to 1$, respectively.

Until time $\tau_1$ the particle behaves as if it were in active transport ($k_2=k_1=0$) and conditionally on the value of $\tau_1$, for any $0\le s\le \tau_1$ the lateral change in position over time $s$, $X^\delta_s-X^\delta_0$, is $\delta n^\delta_s$ where $n^\delta_s$ is a Poisson($rs$) random variable. Moreover, by results of Proposition~\ref{prop:active-transport}, conditionally on the value of $\tau_1$, we have
\[\sup\limits_{0\leq s\leq \tau_1}\big|(X^\delta_s-X^\delta_0)-vs\big|\;\; \mathop{\longrightarrow}\limits_{\delta\to 0}\;\; 0 \; \; \mbox{ a.s.}\]
To get the unconditioned law of $X^\delta_{\tau_1}-X^\delta_0$, we observe that the number of boxes the particle traverses before it steps off $n^\delta_{\tau_1}$ has a Geometric$({k_2}/{(k_2+r)})$ distribution (note that a Poisson rate $r$  process  sampled at an Exponential($k_2$) time independent of the process has this distribution). Since ${k_2}/{\delta(k_2+r)}\mathop{\rightarrow}{k_2}/{v}$, as $\delta\to 0$, $X^\delta_{\tau_1}-X^\delta_0=\delta n^\delta_{\tau_1}$ converges in distribution to $\mbox{Exponential}({k_2}/{v})$ variable. 

Consider now the particle in the non-compartmental model. It is immediate from the definition of the model that $\xi$ has the same law of a continuous-time Markov chain  on $\{0,1\}$ with transition rates $k_2$ and $k_1$, from $1\to 0$ and $0 \to 1$, as $\xi^\delta$. Since we prove our convergence in law results by first conditioning on the values of $\xi^\delta$ and $\xi$, we can without loss of generality henceforth assume $\xi^\delta=\xi$ a.s., and drop its superscript.

Suppose that initially the particle in the non-compartmental model is on transport at $x_0$, so $(X_0,\xi_0)=(x_0,1)$. Conditionally on $\tau_1$, for all $0\le s\le\tau_1$, $X_s-X_0=vs$ and $X_{\tau_1}-X_0= v\tau_1$, hence unconditionally, $X_{\tau_1}-X_0$ is an $\mbox{Exponential}({k_2}/{v})$ random variable. We now have both $\sup_{0\leq s\leq \tau_1}\big|(X^\delta_s-X^\delta_0)-(X_s-X_0)\big|\to 0$ a.s. and $X^\delta_{\tau_1}-X^\delta_0\Rightarrow X_{\tau_1}-X_0$.

Between times $\tau_1$ and $\sigma_1$ the particle in both models stays in place, so conditionally on values of $\tau_1,\sigma_1$, and of $ X^\delta_{\tau_1},X_{\tau_1}$, $\sup_{\tau_1\leq s\leq \tau_1+\sigma_1}\big|(X^\delta_s-X^\delta_{\tau_1})-(X_s-X_{\tau_1})\big|\equiv 0$, and $X^\delta_{\tau_1+\sigma_1}-X^\delta_0\Rightarrow X_{\tau_1+\sigma_1}-X_0$.
At time $\tau_1+\sigma_1$, the same process starts over from initial values $(X^\delta_{\tau_1},1)$ and $(X_{\tau_1},1)$ in the compartmental and non-compartmental model, respectively.

Let $\tau_1, \sigma_1, \tau_2, \dots,$ be the sequence of time increments between the consecutive times when the particle in both models gets off and gets back on transport, let $\sigma_0=0$ and for $i\ge 1$
\[\tau_i=\inf\{t>0: \xi_{\sigma_{_{i-1}}+t}=0\}, \quad \sigma_i=\inf\{t>0: \xi_{\tau_{i}+t}=1\}\]
Then $(\tau_i)_{i\ge 1}$ and $(\sigma_i)_{i\ge 1}$ are independent sequences of i.i.d. Exponential($k_2$) and  Exponential($k_1$) variables, respectively. For any $t>0$, let $\eta_t$ be the number of times the particle in either model gets back on transport  until time $t$, and $\eta'_t$ the number of times it gets off,
\[\eta_t=\inf\{k\ge 0:\sum_{i=1}^k (\tau_i+\sigma_i)\le t\},\quad \eta'_t=\inf\{k'\ge 0:\sum_{i=1}^{k'} \tau_i+\sum_{i=1}^{k'-1}\sigma_i\le t\}\] 
Note that, $\eta'_t=\eta_t\mbox{ iff }\xi_t=1$, and $\eta'_t=\eta_t+1\mbox{ iff }\xi_t=0$. 
Let $\tilde{\tau}_t$ be the last time before time $t$ that the particle changed whether it was on or off transport, that is, $\tau_t=\sup\{0\le s\le t: \xi_{s^-}\neq \xi_s\}$. Then, we have
\[\tau_t=\Bigg\{\begin{tabular}{ll} $\sum\limits_{i=1}^{\eta_t}(\tau_i+\sigma_i)$ & if $\eta'_t=\eta_t,$\\ $\sum\limits_{i=1}^{\eta'_t}\tau_i+\sum\limits_{i=1}^{\eta_t}\sigma_i$& if $\eta'_t=\eta_t+1$.\end{tabular}\]

If $\eta'_t=\eta_t$, then from time $\tau_t$ to $t$ the particle is in active transport, and the same convergence argument as before implies that conditionally on the values of $\tau_t$ and $X^\delta_{\tau_t}$,
\[\sup_{\tau_t\leq s\leq t}|(X^\delta_s-X^\delta_{\tau_t})-v(s-\tau_t)|\to 0\mbox{ a.s.}\] 
Also $X^\delta_{\tau_t}-X^\delta_0=\delta n^\delta_{\tau_t}$ where $n^\delta_{\tau_t}$ is the number of boxes the particle traverses by time $t$. Conditionally on the value of $\eta_t$, $n^\delta_{\tau_t}$  is a sum of $\eta_t$ i.i.d. Geometric$({k_2}/{(k_2+r)})$ random variables, hence $X^\delta_{\tau_t}-X^\delta_0$ converges in distribution to a sum of $\eta_t$ i.i.d. Exponential($k_2/v)$ random variables. In the non-compartmental model, if $\eta'_t=\eta_t$, then conditionally on the values of  $\eta_t$ and $\tau_t$, for $\tau_t\leq s\leq t$, $X_s-X_{\tau_t}= v(s-\tau_t)$ and conditionally only on the value of $\eta_t$, $X^\delta_{\tau_t}-X_0$ is a sum of $\eta_t$ i.i.d. Exponential($k_2/v)$ variables. Hence, conditionally on $\eta_t$ and $\tau_t$, $\sup_{\tau_t\leq s\leq t}\big|(X^\delta_s-X^\delta_{\tau_t})-(X_s-X_{\tau_t})\big|\to 0$ a.s.,  and conditionally only on $\eta_t$, $X^\delta_{\tau_t}-X^\delta_0\Rightarrow X_{\tau_t}-X_0$.

If $\eta'_t=\eta_t+1$, then from time $\tau_t$ to $t$ the particle is in both models stays in place, so conditionally on $\eta_t$ and $\tau_t$, $\sup_{\tau_t\leq s\leq t}\big|(X^\delta_s-X^\delta_{\tau_t}) -(X_s-X_{\tau_t})\big|=0\mbox{  a.s.}$. Also, conditionally only on the value of $\eta'_t$, $n^\delta_{\tau_t}$ is a sum of $\eta'_t$ i.i.d. Geometric$({k_2}/{(k_2+r)})$ variables, hence $X^\delta_{\tau_t}-X^\delta_0$ converges in distribution to a sum of $\eta'_t$ i.i.d. Exponential($k_2/v)$ variables.  
In the non-compartmental model, if $\eta'_t=\eta_t+1$, conditionally only on the value of  $\eta'_t$, $X^\delta_{\tau_t}$ is a sum of $\eta'_t$ i.i.d. Exponential($k_2/v)$ variables. Hence, conditionally only on $\eta'_t$, $X^\delta_{\tau_t}-X^\delta_0\Rightarrow X_{\tau_t}-X_0$.

Now, integrating over the possible values of $\eta_t$,$\eta'_t$ and $\tau_t$, we get that for any $t\ge 0$, $(X^\delta_t,\xi^\delta_t)\Rightarrow(X_t,\xi_t)$. Convergence of finite dimensional distributions follows from an iterative use of the Markov property of $(X^\delta,\xi^\delta)$ and $(X,\xi)$, and the fact that the increments of both $(X^\delta,\xi^\delta)$ and $(X,\xi)$ are stationary.

In order to verify tightness, Theorem 16.11 \cite{Kal}, of the sequence of Markov processes $\{(X^\delta,\xi^\delta)\}_{\delta >0}$, because $(X^\delta,\xi^\delta)$ has stationary increments and is strong Markov,  it will suffice to check that for any $\epsilon>0$
\[\lim_{h\to 0}\limsup_{\delta\to 0}\p{||(X^\delta_h,\xi^\delta_h)-(X^\delta_0,\xi^\delta_0)||>\epsilon}= 0\]
where $||(x_1,\xi_1)-(x_2,\xi_2)||=|x_1-x_2|+|\xi_1-\xi_2|$ is a distance metric on $\mathbf{R}_+\times\{0,1\}$. 
For any $\delta>0$, the first change in the continuous-time Markov chain $\xi^\delta$ happens after an Exponential($k$)  time (where $k=k_2$ or $k=k_1$ depending on whether $\xi^\delta_0=1$ or $\xi^\delta_0=0$), and is independent of $\delta$. Hence, at time $h$ later, $\p{\xi^\delta_h\neq  \xi^\delta_0}\le 1-e^{-kh}$. Moreover, irrespective of the value of $\xi^\delta$, at time $h$ later the value of $|X^\delta_h-X^\delta_0|\le \delta n^\delta_h$ where $n^\delta$ is a Poisson process with rate $r=v/\delta$. Hence,  $\p{|X^\delta_h-X^\delta_0|>\epsilon}\le \delta \E{n^\delta_h}/\epsilon=vh/\epsilon$. 
Combining the two gives \[\p{||(X^\delta_h,\xi^\delta_h)-(X^\delta_0,\xi^\delta_0)||>\epsilon}\le 1-e^{-kh}+vh/\epsilon \mbox{ for any }\delta>0,\] and the desired limit follows.
\end{proof}

\

The process $(X^\delta,\xi^\delta):t\in[0,\infty)\mapsto(X^\delta_t,\xi^\delta_t)\in\delta\mathbb{Z}_+ \times \{0,1\}$  is a Markov process with cadlag paths whose generator is given by 
\begin{align*}A^\delta f(x,\xi)&=r\xi\big[f(x+\delta,\xi)-f(x,\xi)\big]\\& + k_2\xi\big[f(x,\xi-1)-f(x,\xi)\big]+ k_1(1-\xi)\big[f(x,\xi+1)-f(x,\xi)\big]\end{align*}
for all  $f\in\mathcal{D}(A^\delta)=\mathcal{C}^0(\delta\mathbb{Z}_+ \times \{0,1\})$.

 The piecewise linear process $(X,\xi):t\in[0,\infty)\mapsto(X_t,\xi_t)\in\mathbb{R}_+ \times \{0,1\}$ is a Markov process with continuous paths whose generator is the closure of the operator
\begin{align*}Af(x,\xi)&={v}\mathbf\xi\partial_x f(x,\xi)\\& + k_2\xi\big[f(x,\xi-1)-f(x,\xi)\big]+k_1(1-\xi)\big[f(x,\xi+1)-f(x,\xi)\big]\end{align*}
for all  $f\in\mathcal{D}(A^\delta)=\mathcal{C}^{1,0}(\mathbb{R}_+ \times \{0,1\})$.

Letting $\iota^\delta:\delta\mathbb{Z}_+\times \{0,1\}\mapsto\mathbb{R}_+\times \{0,1\}$ be an embedding, and $f^\delta=f\circ\iota^\delta$, then $A^\delta f^\delta\to Af$ as $\delta\to 0$ for all $f\in\mathcal{C}^{1,0}(\mathbb{R}_+ \times \{0,1\})$ imply that the finite dimensional distributions of $(X^\delta,\xi^\delta)$ converge to those of $(X,\xi)$. Verification of additional conditions,  see Theorem 19.25 of \cite{Kal}, would also imply convergence of processes with generators $\{\big(A^\delta,\mathcal{D}(A^\delta)\big)\}_{\delta>0}$ to the process with generator $\big(A,\mathcal{D}(A^\delta)\big)$, however, we thought this way of showing convergence in law was not as instructive.

\

The fact that an individual particle will have the distribution given by Proposition~\ref{prop:on-off} as the size of the boxes decreases means that our model is a microscopic version of the stochastic model used by Brooks \cite{B99}, and that the hydrodynamic limit of our model as $\delta\to 0$ is equal to the macroscopic stochastic model from \cite{B99}. 

An approximation of the particle's position $X_t$ is obtained in \cite{B99} to be $X_t\approx\mu t +\sqrt{t}\sigma Z$ as $t\to\infty$, where $\mu=k_1v/(k_2+k_1), \sigma=2k_2k_1v^2/(k_2+k_1)^3$  and $Z$ is a standard Normal variable. That approximation is valid only for large fixed values of $t$, while we next extend this result to give an approximation for the whole time trajectory of the particle's path. This is accomplished by the following  functional central limit theorem for the position of the particle undergoing stochastic transport.

\begin{proposition}\label{prop:wavefront} 
Let $X$ be the position of a particle following the piecewise linear Markov process from Proposition~\ref{prop:on-off} started at  $X_0=0$ on transport, then
\[\sup\limits_{s\leq t}\big|\frac{X_{ns}}{n}-\frac{k_1}{k_2+k_1}vs\big|\;\; \mathop{\longrightarrow}\limits_{n\to \infty}\;\; 0 \; \; \mbox{ a.s.} \;\; \forall t>0\] 
\mbox{and}, if  $B$  denotes a standard Brownian motion
\[\sqrt{n}\Big(\frac{X_{nt}}{n}-\frac{k_1v}{k_2+k_1}t\Big)_{t\geq 0}\;\;
\mathop{\Longrightarrow}\limits_{n \to \infty}\;\; \sqrt{\frac{2k_2k_1v^2}{(k_2+k_1)^3}} \big(B_t\big)_{t\geq
0}\]
in distribution on the space of continuous functions.
\end{proposition}

\

\begin{proof} For these results we use the notion of stochastic averaging \cite{Kh66a}\cite{Kh66b}\cite{k92}. 
Note that the indicator process $\xi$ for being on- or off- transport is independent of the position $X$ of the particle. Hence, the position of the particle $X$ is a linear random evolution process, Ch 12 of \cite{ek86}, \cite{GH69}, driven by the independent indicator process $\xi$. The generator of $(X,\xi)$ is the closure of the operator
\[Af(x,\xi)=\sigma(\xi)\partial_x f(x,\xi) +\lambda(\xi)\big[f(x, s(\xi))-f(x,\xi)\big]\]
for all  $f\in\mathcal{D}(A)=\mathcal{C}^{1,0}_0(\mathbb{R}_+ \times \{0,1\})$ (the space of all continuously differentiable functions in $x$ continuous in $\xi$ and vanishing at infinity),
where \[\sigma(\xi)=v\xi, \;\; \lambda(\xi)=k_2\xi +k_1(1-\xi),\; \mbox{ and }\;s(\xi)=1-\xi\] (when $\xi=1$: $\sigma=v, \lambda=k_2, s=0$ and when  $\xi=0$: $\sigma=0, \lambda=k_1, s=1$).
In other words, if $(X_0, \xi_0)=(0,1)$ we have
\[X_t=\int_0^t v\xi_s ds, \quad \xi_t=\frac{1}{2}\big(1+(-1)^{Y(\int_0^t \lambda(\xi_s)ds)}\big)\]
where $Y$ is a rate $1$ Poisson process, and $Y(\int_0^t \lambda(\xi_s)ds)$ is a counting process of the number of switches of $\xi$ until time $t$.

Rescaling time and position of the process by $1/n$, we get that $(\frac{X_{n\cdot}}{n}, \xi_{n\cdot})$ satisfies
\[\frac{X_{nt}}{n}=\int_0^t v\xi_{ns} ds, \quad \xi_{nt}=\frac{1}{2}\big(1+(-1)^{Y(n \int_0^t \lambda(\xi_{ns})ds)}\big)\]
and its generator is
\[A^nf(x,\xi)=v\xi\partial_x f(x,\xi) +n(k_2\xi+(k_1(1-\xi)))\big[f(x, s(\xi))-f(x,\xi)\big]\]
for all  $f\in\mathcal{D}(A^n)=\mathcal{C}^{1,0}_0(\mathbb{R}_+ \times \{0,1\})$. 
Note that $\xi_{n\cdot}$ switches at rate proportional to $n$, forming an ergodic Markov chain with stationary  distribution $\pi(1)={k_1}/{(k_2+k_1)}, \pi(0)={k_2}/{(k_2+k_1)}$, and $\int \sigma(\xi)\pi(\xi)=v\int\xi\pi(\xi)=v{k_1}/{(k_2+k_1)}$. Hence, the strong ergodic theorem implies that
 \[\frac{X_n}{n}=\frac1n\int_0^{n}v\xi_s ds\;\mathop{\longrightarrow}\limits_{ n\to \infty}\;\frac{vk_1}{k_2+k_1} \; \mbox{ a.s.}\] 
We can extend this to a functional statement on any finite time interval $[0,t]$. 
Fix $t>0$ and take any $\Delta>0$, then there exists $n_\Delta<\infty$ a.s. such that \[\big|\frac{X_n}{n}-\frac{vk_1}{k_2+k_1} \big|<\frac\Delta t,\quad  \mbox{  for }\forall n>n_\Delta\]
Now, let $M=\sup_{n\le n_\Delta} |X_n-k_1/(k_2+k_1) vn|$, which is finite a.s. since $n_\Delta<\infty$ a.s. Let $n>M/\Delta$. Then, for any $0\le s\le t$ we have that either $ns>n_\Delta$ in which case 
\[\big|\frac{X_{ns}}{n}-\frac{vk_1}{k_2+k_1}s \big|=\big|\Big(\frac{X_{ns}}{ns}-\frac{vk_1}{k_2+k_1}\Big)s \big|<\frac\Delta t s\leq \Delta\]
or, $ns\le n_\Delta$ in which case 
\[\big|\frac{X_{ns}}{n}-\frac{vk_1}{k_2+k_1}s \big|=\frac 1n\big|X_{ns}-\frac{vk_1}{k_2+k_1}ns \big|<\frac M n< \Delta\]
implying that we have $\sup_{0\le s\le t}\big|\frac1n{X_{ns}}-{vk_1}/{(k_2+k_1)}s \big|<\Delta$ whenever $ n>M/\Delta$. 
\\

Once we rescale the position for the particle by $1/\sqrt{n}$ and time by $1/n$,  $\xi$ still changes at a much faster rate than the position of the particle $X$. The generator of the rescaled centered process $(X^n,\xi^n)$ defined as
\[X^{n}_t:=\sqrt{n}\big(\frac{X_{nt}}{n}-\frac{k_1v}{k_2+k_1}t\big), \quad \xi^n_t:=\xi_{nt}\]  
is the closure of the operator
\[\bar{A}^{n}f(x,\xi)=\big(\sqrt{n}\sigma(\xi)-\frac{k_1v}{k_2+k_1}\big)\partial_x f(x,\xi) +n\lambda(\xi)\big(f(x, s(\xi))-f(x,\xi)\big)\]
on $f\in \mathcal{D}(A^n)=\mathcal{C}^{1,0}_0(\Bbb R\times\{0,1\})$. We will use the stochastic averaging theorem (Theorem 2.1 of \cite{k92}) to show that the paths of the centered rescaled process converge in distribution to paths of a Brownian motion with a diffusion coefficient equal to $2k_2k_1/(k_2+k_3)^3v^2$.

Let $h(\xi)$ be the function 
\[h(\xi)=v\frac{k_2k_1}{(k_2+k_1)^3}\frac{1}{\lambda(\xi)}=v\frac{k_2k_1}{(k_2+k_1)^3}\frac{1}{k_2\xi+k_1(1-\xi)}\]
Then $h(1)=v{k_1}/{(k_2+k_1)^2}$, $h(0)=-v{k_2}/{(k_2+k_1)^2}$ imply that $h(s(1))-h(1)=-{v}/{(k_2+k_1)}$, $h(s(0))-h(0)={v}/{(k_2+k_1)}$, which in turn imply that $\lambda(1)\big(h(s(1))-h(1)\big)=-{vk_2}/{(k_2+k_1)}$, $\lambda(0)\big(h(s(0))-h(0)\big)={vk_1}/{(k_2+k_1)}$, so that  for $\xi\in\{0,1\}$
\[\lambda(\xi)\big(h(s(\xi))-h(\xi)\big)=-\big(v\xi-v\frac{k_1}{k_2+k_1}\big)\]
Now, for any $f\in\mathcal{C}^{2}_0(\Bbb R)$ define a sequence of functions $f^{n}\in\mathcal{C}^{1,0}_0(\Bbb R\times\{0,1\})$ by 
\[f^{n}(x,\xi)=f(x)+\frac{1}{\sqrt{n}}h(\xi)\partial_xf(x)\] 
Then $f^{n}\to f$ as $n\to\infty$ and 
\begin{align*}
\bar{A}^{n}f^{n}(x,\xi)\;=&\;\sqrt{n}\big(v\xi -v\frac{k_1}{k_2+k_1}\big)\partial_xf(x)+\big(v\xi -v\frac{k_1}{k_2+k_1}\big)h(\xi)\partial^2_xf(x)\\&\;+\sqrt{n}\lambda(\xi)\big(h(s(\xi))-h(\xi)\big)\partial_xf(x)\\
\;=&\;\big(v\xi -v\frac{k_1}{k_2+k_1}\big)h(\xi)\partial^2_xf(x)
\;=\;\bar{A}f(x)
\end{align*}
where $\bar{A}$ is defined on $\mathcal{D}(\bar{A})=\mathcal{C}^{2}_0(\Bbb R)$ by
\[\bar{A}f(x)=\frac{k_2k_1v^2}{(k_2+k_1)^3}\partial^2_xf(x)\]
Define a sequence of processes \[\ep^{f,n}_t=\frac{1}{\sqrt{n}}h(\xi^n_t)\partial_xf(X^n_t)=f^n(X^n_t,\xi^n_t)-f(X^n_t)\] Then our earlier calculation implies that for any $f\in\mathcal{D}(\bar{A})$ 
\begin{align*}f(X^n_t)-\int_0^t \bar{A}f(X^n_s)ds +\ep^{f,n}_t=f^n(X^n_t,\xi^n_t)-\int_0^t A^nf(X^n_s,\xi^n_s)ds\end{align*}
is a sequence of martingales. Since $f\in\mathcal{C}^{2}_0(\Bbb R)$, $\xi^n_t\in\{0,1\}$ it is clear that
\[\sup_{n}\E{\int_0^t\big|\bar{A}f(X^n_s)\big|^2ds}<\infty \;\mbox{ and}\quad \E{\sup_{s\le t}\big|\ep^{f,n}_s\big|}\mathop{\longrightarrow}\limits_{n\to\infty} 0\]
In order to apply Theorem 2.1 of \cite{k92} on stochastic averaging it is only left to show the process $X^n$ satisfies the compact containment condition, that is for any $t>0$ and $\Delta>0$ there exists a compact set $K\subset\Bbb R$ such that
\[\inf_{n}\p{X^n_s\in K \;\forall s\le t }\ge 1-\Delta\]
This follows from the fact that $X^n_t+h(\xi^n_t)/\sqrt{n}$ is a sequence of martingales (let $f(x)=x$) with mean 
\[\E{X^n_t+\frac{h(\xi^n_t)}{\sqrt{n}}}=X^n_0+\frac{h(\xi^n_0)}{\sqrt{n}}=\frac{k_2k_1}{(k_2+k_1)^3}\frac{v^2}{\sqrt{n}}\]
 and second moment  (let $f(x)=x^2$)
 \[\E{\Big(X^n_t+\frac{h(\xi^n_t)}{\sqrt{n}}\Big)^2}=\frac{k_2k_1}{(k_2+k_1)^3}v^22t+\frac{\E{h(\xi^n_t)}}{n}\] 
 So, by Doob's inequality
\[\p{\sup_{s\le t}\Big|X^n_s+\frac{h(\xi^n_t)}{\sqrt{n}}\Big|\ge C}
\le \frac{4}{C^2}\E{\Big(X^n_t+\frac{h(\xi^n_t)}{\sqrt{n}}\Big)^2}=\frac{4}{C^2}\Big(\frac{k_2k_1}{(k_2+k_1)^3}v^22t+\frac{\E{h(\xi^n_t)}}{n}\Big)\] 
Noting that $h_{\min}:=v\min(k_2,k_1)/(k_2+k_1)^2\le h\le h_{\max}:=v\max(k_2,k_1)/(k_2+k_1)^2$, and choosing $C$ (given on $t$ and $\Delta$) so that the right hand side of the inequality with $n=1$ is less than $\Delta$, shows that with $K=[-C-h_{\max},C-h_{\min}]$ the compact containment condition holds for $(X^n)_{n\ge 1}$.

Now Theorem 2.1 of \cite{k92} implies that $X^n\Rightarrow W$  in distribution on the Skorkhod space of continuous functions, where $W$ is a process with generator $\bar{A}$, and consequently has the same distribution as $\;\sqrt{{2k_2k_1v^2}/{(k_2+k_1)^3}}B\;$ where $B$ is a standard Brownian motion. \end{proof}

\subsection{Connection to Partial Differential Equations Models}
\label{subsec:particle-connect-pde}

In order to demonstrate the connection between our model and the
PDEs seen in \cite{reed90}\cite{reed94}\cite{friedman07}, consider the process $(X,\xi)$ of the particle following the piecewise linear Markov process, and for any $x\ge 0, t\ge 0$ let
\begin{align*} q(x,t)=\p{X_t\in dx, \xi_t=1}/dx,\quad p(x,t)=\p{X_t\in dx, \xi_t=0}/dx\end{align*}  
denote the probability densities of the particle's location $x$ on and off transport, respectively, over time.
Kolmogorov forward equations for $(X,\xi)$ imply that $q$ and $p$ satisfy
the system of PDEs
\begin{align}
\label{eq:pde-q} \partial_t q(x,t) + v \partial_x
q(x,t) &= -k_2 q(x,t) + k_1 p(x,t) \\
\label{eq:pde-p} \partial_t p(x,t) &= k_2 q(x,t) - k_1 p(x,t).
\end{align}

When $k_2 = 0 = k_1$, the limiting PDE is simple linear transport:
$(\partial_t + v\partial_x)q(x,t)=0.$  The  initial condition
$q(x,0) = \delta_0(x)$ corresponds to the density of a single
particle at the origin at $t = 0.$ The time evolution via simple
linear transport is translation of the delta function, while the
time evolution via the equations \eqref{eq:pde-q} and
\eqref{eq:pde-p} will have a spreading profile. This is clear from the macroscopic limits of $(X^\delta,\xi^\delta)$ as $\delta\to 0$. When $k_2=k_1=0$ the particle never switches off from traveling on transport at speed $v$ and is deterministic, as seen in Proposition~\ref{prop:active-transport}. When $k_2,k_1>0$ the particle follows a truely stochastic  process $(X,\xi)$ with a non-zero variance, as seen in Proposition~\ref{prop:on-off} and Proposition~\ref{prop:wavefront}. 

In the experiments described in the introduction one sees
``approximate'' traveling waves of radioactivity in the axons in the
sense that there is a slowly spreading wave front moving at constant
velocity away from the cell body. Equations \eqref{eq:pde-q} and \eqref{eq:pde-p} are
linear and do not have solutions that are bounded traveling waves.
It was shown by a perturbation theory argument in
\cite{reed90}\cite{reed94} that as $\ep \rightarrow 0$ the solutions of
\begin{align} \label{eq:pde-q-ep}
\ep (\partial_t + v \partial_x)
q^{\ep}(x,t) &= -k_2 q^{\ep}(x,t) + k_1 p^{\ep}(x,t), \hspace{.4in}  \\
\label{eq:pde-p-ep} \ep \partial_t p^{\ep}(x,t) &= k_2 q^{\ep}(x,t)
- k_1 p^{\ep}(x,t).
\end{align}
subject to $q^{\ep}(0,t) = q_0$, are to leading order
\[  q^{\ep}(x,t) = c_1H(\frac{x-\mu t}{\ep^{1/2}},t), \hspace{.5in} p^{\ep}(x,t) = c_2H(\frac{x-\mu t}{\ep^{1/2}},t), \]
where $H$ satisfies the heat equation
\begin{eqnarray}     
\label{eq:heat}
\partial_sH(y,s) =  \frac{\sigma^2}{2}\partial_{yy}H(y,s), \hspace{.5in} H(y,0) = \chi_{(-\infty,0)} 
\end{eqnarray}
\begin{eqnarray}
\mu = \frac{k_2v}{k_2 + k_1},  \hspace{.1in}  \sigma^2 = \frac{2k_2k_1v^2}{(k_2 + k_1)^3}
& \hspace{.5in}  c_1=\frac{k_1}{k_2+k_1}, & \hspace{.1in}  c_2=\frac{k_2}{k_2+k_1}.\nonumber
\nonumber
\end{eqnarray}
This asymptotic form is valid for small $\ep$, that is for large
$k_2$ and $k_1$. However, if we set $q(x,t)= q^{\ep}(\frac{x}{\ep},
\frac{t}{\ep})$ and $p(x,t)= p^{\ep}(\frac{x}{\ep}, \frac{t}{\ep})$,
then $q$ and $p$ satisfy \eqref{eq:pde-q-ep} and \eqref{eq:pde-p-ep}, so the solutions of \eqref{eq:pde-q} and
\eqref{eq:pde-p} behave like approximate traveling waves for large $t$ and large
$x$ whether or not $k_2$ and $k_1$ are large. These results have been proven rigorously
by Friedman and coworkers
\cite{friedman05}\cite{friedman06}\cite{friedman07}\cite{friedman07-2}. 

To see that our Proposition \ref{prop:wavefront} provides another rigorous proof of these properties, albeit using stochastic methods, consider the process $(\ep X_{\cdot/\ep},\xi_{\cdot/\ep})$ (that is $(\frac1n{X_{n\cdot}},\xi_{n\cdot})$ with $n=1/\ep$), and let 
\begin{align*} q^\ep(x,t)=\p{\ep X_{t/\ep}\in dx, \xi_{t/\ep}=1}/dx,\quad p^\ep(x,t)=\p{\ep X_{t/\ep}\in dx, \xi_{t/\ep}=0}/dx\end{align*}  
be the probability densities for this process. The generator of this process is $A^{n}$ ($n=1/\ep$), so the Kolmogorov forward equations  imply that $q^\ep$ and $p^\ep$ satisfy the system of PDEs \eqref{eq:pde-q-ep} and \eqref{eq:pde-p-ep}. 
Our result from  Proposition \ref{prop:wavefront} states that
\[\p{\ep X_{t/\ep} \in dx}/dx\approx H(\frac{x-\mu t}{{\ep}^{1/2}},t)\] 
for small $\ep>0$, where $H$ satisfies \eqref{eq:heat}. 
Hence, $q^\ep+p^\ep \approx H(\frac{x-\mu t}{{\ep}^{1/2}},t)$, and $\p{\xi_{t/\ep}=1}\approx k_1/(k_2+k_1)$, $\p{\xi_{t/\ep}=0}\approx k_2/(k_2+k_1)$, gives the result that $q^\ep(x,t)$  and $p^\ep(x,t)$ are  well approximated by $\frac{k_1}{k_2+k_1}H(\frac{x-\mu t}{\ep^{1/2}},t)$ and $\frac{k_2}{k_2+k_1}H(\frac{x-\mu t}{\ep^{1/2}},t)$. 

\

\section{Dynamics from the spatial system perspective}
\label{sec:system-persective}

\subsection{The spatial system in equilibrium}
\label{subsec:system-stationary-dist}

We are now ready to characterize the steady state dynamics induced
by continually adding particles from the nucleus and removing them
when they reach the distal end of the cell.

\begin{proposition} \label{prop:equilib} Let  $(\{(Q_i(t), P_i(t)), i =1, \dots, N\})_{t\ge 0}$ be the number of particles in the axonal transport system with compartments of size $\delta$, on and off transport, respectively.
Suppose the rate of production of particles from the source is $rq_0=v\rho_0$. Then this Markov chain has the product-form stationary distribution
\[(Q_i,P_i) \sim Pois(q_0)\otimes Pois\left(\frac{k_2 q_0}{k_1}
\right)\]
where all $\{(Q_i,P_i),i=1,\dots,N\}$ are mutually independent.
\end{proposition}

\begin{proof}
Since the production rate is $rq_0$, the generator of the process $(Q_1,P_1)$ is
\begin{align*}
\gen_{q_0} f(q,p) &= [f(q+1,p) - f(q,p)] rq_0  + [f(q-1,p) - f(q,p)] r q\\
& + [f(q-1,p+1) - f(q,p)] k_2 q  + [f(q+1,p-1) - f(q,p)] k_1 p
\end{align*}
If we use $f(q,p) = Q_1(t)$ and $f(q,p) = P_1(t)$, and take expectations, we get a system of ODEs governing the change in $\E{Q_1}, \E{P_1}$ over time
\begin{align*}
\frac{d\E{Q_1}(t)}{dt}&=rq_0 - r\E{Q_1(t)} - k_2\E{Q_1(t)} + k_1\E{P_1(t)}\\
\frac{d\E{P_1}(t)}{dt}&= k_2\E{Q_1(t)} - k_1\E{P_1(t)}
\end{align*}
indicating that in equilibrium  in the first section the mean numbers of on-transport particles and off-transport particles are $\E{Q_1}=q_0$, and $\E{P_1}=\E{Q_1} \frac{k_2}{k_1}=q_0 {k_2}/{k_1}$, respectively. 
Let $\pi_{q_0}(q,p)=\pi_{\lambda_Q}(q)\otimes\pi_{\lambda_P}(p)$
be a product of two independent Poisson distributions with rates
$\lambda_Q=q_0$, and
$\lambda_P=q_0{k_2}/{k_1}$ respectively. To show that
$\pi_{q_0}(q,p)$ is a stationary distribution for the process
$(Q_1,P_1)$ we need to check that
$$
\sum_{q=0}^\infty\sum_{p=0}^\infty  \gen_{q_0}
f(q,p) \pi_{q_0}(q,p) =0
$$
for any choice of function
$f\in\mathcal{D}(\gen_{q_0})$.
\begin{align*}
&\sum_{q=0}^\infty\sum_{p=0}^\infty  \gen_{q_0} f(q,p) e^{-(\lambda_Q+\lambda_P)} \frac{\lambda_Q^q}{q!}\frac{\lambda_P^p}{p!}\\
&\quad = e^{-(\lambda_Q+\lambda_P)} \sum_{q=0}^\infty\sum_{p=0}^\infty  \frac{\lambda_Q^q}{q!}\frac{\lambda_P^p}{p!}
\Big( [f(q + 1,p) - f(q, p)] r{q_0}  + [f(q - 1, p) - f(q, p)] r q\\
& \qquad\qquad +  [f(q - 1, p + 1) - f(q, p)] k_2 q  + [f(q + 1, p -
1) - f(q, p)] k_1 p\Big)
\end{align*}
In the two sums the factor
multiplying $f(q,p)$ for any $(q,p)\in\mathbb{N}\times\mathbb{N}$ comes only from
terms involving $\{q-1,q,q+1\}$ and $\{p-1,p,p+1\}$ and equals
$e^{-(\lambda_Q+\lambda_P)} $ times
\begin{align*}
&\frac{\lambda_Q^{q-1}}{(q-1)!}\frac{\lambda_P^p}{p!}rq_0
-\frac{\lambda_Q^q}{q!}\frac{\lambda_P^p}{p!}rq_0
+\frac{\lambda_Q^{q+1}}{(q+1)!}\frac{\lambda_P^p}{p!}r(q+1)
-\frac{\lambda_Q^{q}}{q!}\frac{\lambda_P^p}{p!}rq \\
&\quad\qquad\qquad+ \frac{\lambda_Q^{q+1}}{(q+1)!} \frac{\lambda_P^{p-1}}{(p-1)!}k_2(q+1)
-\frac{\lambda_Q^{q}}{q!}\frac{\lambda_P^p}{p!}k_2q\\
&\quad\qquad\qquad+\frac{\lambda_Q^{q-1}}{(q-1)!}\frac{\lambda_P^{p+1}}{(p+1)!}k_1(p+1)
-\frac{\lambda_Q^{q}}{q!}\frac{\lambda_P^p}{p!}k_1p\\
&=\frac{\lambda_Q^q}{q!} \frac{\lambda_P^p}{p!}
\Big( \frac{q}{\lambda_Q}rq_0-rq_0+\frac{\lambda_Q}{q+1}r(q+1)-rq+\frac{\lambda_Q}{q+1}\frac{p}{\lambda_P}k_2(q+1)\\ &\quad\qquad\qquad-k_2q+\frac{q}{\lambda_Q}\frac{\lambda_P}{p+1}k_1(p+1)-k_1p\Big)\\
&=\frac{\lambda_Q^q}{q!} \frac{\lambda_P^p}{p!}\Big(\frac{q}{q_0}rq_0-rq_0+q_0 r-rq\Big)=0
\end{align*}
since  $\lambda_Q=q_0$ and
${\lambda_P}/{\lambda_Q}={k_2}/{k_1}$.

Thus, in equilibrium the input rate for $(Q_2,P_2)$, which is $rQ_1$, has a Poisson distribution with mean $rq_0$, and is independent of $P_1$. Let $\pi_{q_0}(q_1)\otimes\pi_{\lambda_Q}(q_2)\otimes\pi_{\lambda_P}(p_2)$ be a product of three Poisson distributions with rates $q_0$, $\lambda_Q=q_0$, and $\lambda_P=q_0{k_2}/{k_1}$ respectively. To show that this is a stationary distribution for the process
$(Q_1,Q_2,P_2)$ we need to check that
$$
\sum_{q_1=0}^\infty\sum_{q_2=0}^\infty\sum_{p_2=0}^\infty  \gen_{q_1}
f(q_2,p_2) \pi_{q_0}(q_1)\pi_{q_1}(q_2,p_2) =0
$$
for any choice of function
$f\in\mathcal{D}(\gen_{q_1})$.
For each fixed value $q_1$, according to our previous calculation the inner two sums give 0, so  the whole sum is $0$. 

Thus, in equilibrium, $(Q_2,P_2)$ have the distribution $\pi_{\lambda_Q}(q)\otimes\pi_{\lambda_P}(p)$ with $\lambda_Q=q_0, \lambda_P=q_0{k_2}/{k_1}$ as well, and are independent from $(Q_1,P_1)$.
It follows by induction that the stationary distributions for
$\{(Q_i,P_i)\}$ are independent and identically distributed as
$\pi_{\lambda_Q}(q)\otimes\pi_{\lambda_P}(p)$, with $\lambda_Q=q_0,\lambda_P=q_0
k_2/k_1$. This is also an example of a clustering process satisfying
the detailed balance conditions with linear rates discussed in Sec.
8.2 of \cite{K79}.
\end{proof}
 
We point out that the mean number of particles both on and off transport is $(1+\frac{k_2}{k_1})q_0=(1+\frac{k_2}{k_1})\rho_0\delta$, scaling with the size of a compartment. To obtain the mean number of particles per unit length we  add particles in $\approx 1/\delta$ compartments, and the mean number of particles per unit length is $(1+\frac{k_2}{k_1})\rho_0$ independent of the choice of compartment size.

One immediate consequence is the analogous result for the number of particles in the stochastic non-compartmental model  at any location along the axon. Namely, suppose the particles move according to the piecewise-linear Markov process $(X,\xi)$ with a Poisson rate $\rho_0$ influx of new particles at location $0$.  Then, at any location $0<x<L$ along the axon, the numbers of particles $(Q_{(x,x+dx)}, P_{(x,x+dx)})$ on and off transport, respectively, have the stationary distribution $\mbox{Pois}(\rho_0dx)\otimes\mbox{Pois}(\frac{k_2}{k_1}\rho_0dx)$ where for any $x_1,\dots x_k\in(0,L)$, $\{Q_{x_i}\}_{1\le i\le k}$ and $\{P_{x_i}\}_{1\le i\le k}$ are mutually independent.
We note that this result would not have been obvious to notice without going through the compartmental model first, yet its consequences for prediction and analysis of the long term stochastic behavior of the system are quite powerful.

\subsection{Homogeneity of the axons at equilibrium}
\label{subsec:homog-at-equilib}

Recall that $\delta = 10nm$, roughly the step size of motor
proteins, and that axons can be up to one meter in length. Thus we
are interested in phenomena on all the length scales
$10^{\nu}\delta$, where $\nu = 0, 1, 2, \ldots 8.$ Let $\Delta =
10^{\nu}\delta;$ we want to determine how similar different segments
of the axon of size $\Delta$ are. Let $Q_{\Delta}$ and $P_{\Delta}$
denote the numbers of on-track and off-track particles in a segment
of length $\Delta.$

In equilibrium, $Q_{\Delta}$ and $P_{\Delta}$  are both  sums of
$10^\nu$ independent Poisson random variables with parameters
$\lambda_Q=q_0$ and $\lambda_P=\frac{k_2}{k_1}q_0$, respectively.
Therefore the distributions of $Q_{\Delta}$ and $P_{\Delta}$ are
Poisson with parameters $ 10^\nu \lambda_Q$ and $10^\nu \lambda_P$,
respectively. The mean and the variance of the number of particles
in the segment of length $\Delta$ is  $ 10^\nu (\lambda_Q +
\lambda_P).$ To see how homogeneous different slices of length
$\Delta$ are, we consider the coefficient of variation,
$c_{\Delta}$, which is the standard deviation divided by the mean.
$$
c_{\Delta} = \frac{1}{\sqrt{(\lambda_Q + \lambda_P) 10^{\nu}}}=
\frac{1}{\sqrt{(1+k_2/k_1)q_0 10^{\nu}}}.
$$

As indicated in Section \ref{subsec:model}, $q_0$ is in the range
0.25 to 2.5 in different axons. For illustrative purposes here, we
will assume $q_0 = 1$. Since $k_2/k_1 = 3$, we see that the
scale-dependent coefficient of variation $ c_\Delta = 1 / (2
\sqrt{10^{8} \Delta})$. Therefore at the ten nanometer scale the
coefficient of variation is simply 1/2.  At the micron scale
$c_\Delta = 1 / 20$ and at the millimeter scale $c_\Delta = 0.5
\times 10^{-5/2}$.  The cutoff between ``high variance'' and ``low
variance'' distributions is usually considered to be when the
coefficient of variation is near 1, so by this standard the axon is
extremely homogeneous in its length at large scales.\\

\subsection{Balance between efficient transport and targeted delivery}
\label{sec:target}

The preceding characterization of the transport apparatus enables us to address questions concerning whole cell function. One core issue is that intracellular transport must simultaneously accommodate two functional demands:  some material, such as the enzymes used to construct neurotransmitters, must be transported from the soma to the axon terminal in a timely manner; whereas other cargo, such as sodium channels, need to be delivered to unspecified locations as needs arise throughout the length of the cell.  The tradeoff between these two goals is clear.  If a typical vesicle spends the vast majority of its time in transport mode, the mean velocity will be close to the mean on-transport velocity, but any needs that arise in the central part of the cell will be neglected.  On the other hand, if a typical vesicle spends too much time off transport, presumably available for use if needed locally, then it will take substantially longer to traverse the entire cell. 

Recently Bressloff and Newby \cite{newby09,newby10} modeled particles that are created near the nucleus that then undergo intermittent search (being in search mode while off transport and not in search mode when on transport) for a target hidden somewhere along the axon. Their model is a non-compartmental individual particle model with the additional feature that vesicles can move backwards as well as forwards. They compute the probability that the particle is successful and conditioning on success, the mean first hitting time. 
With our system-wide model we can accommodate the observation the if a given vesicle misses the target, another vesicle with similar cargo will pass by before too long. We will assess this hitting time under two assumptions about the density of relevant material.  Our standing assumption $q_0 \sim 1$ is appropriate for types of cargo that are found densely throughout the cell. In this setting, the wait time is essentially just the time it takes for one of several nearby vesicles to unbind from transport in the target region. A more interesting case is a setting where the needed cargo is sparsely distributed, say $q_0 \sim 10^{-3}$.  In this setting, if the first cargo to reach the target region fails to unbind, there will be significant time before the next arrival.  As we will see, we can still assess the trade-off intrinsic between risking a target miss and diminishing the time of the next arrival. 

To make the discussion concrete we define a target region $R_n = \{i_* + 1, \dots, i_* + n\}$ where $i_* \in 1, \ldots, N - n$.  At time zero, we take the system $\{Q_i(0), P_i(0)\}$ to be drawn from the stationary distribution described by Proposition \ref{prop:equilib} conditioned on the event that for all $i \in R_n$, $Q_i(0) = P_i(0) = 0$. We introduce the hitting time $H_n := \inf \{t > 0: \sum_{i \in R_n} P_{i}(t) > 0 \}$, which marks the first time a particle is off transport while in $R_n$.  Since computing the mean of $H_n$ is analytically intractable, we introduce another hitting time $H_n'$, stochastically dominating $H_n$, that nevertheless reflects the essential tradeoff between maximizing mean velocity and making detachment from transport likely in the target region.  

Let $I_* := \{ i \in \{1, \ldots, i_*\} \, | \, Q_i(0) + P_i(0) > 0 \}$ be the set of all non-empty sections of the cell at time 0. Among the particle in these sections some will be ``successful'' in that they will detach from transport in the target region, while others will be unsuccessful.  Labeling each particle in these terms, we decompose $I_0$ into locations with successful particles $I^s_*$ and unsuccessful particles $I^u_*$. We are interested in the time $H'_n$ at which the rightmost successful particle, that is, a particle starting from position $i_m = \max\{I_*^s\}$, detaches while in $R_n$. If $S=0$ then we define $i_m=\max\{\emptyset\}:=0$. For the position of this particle we use the notation $(X^\delta_t)_{t\ge 0}$ from Section \ref{sec:particle-perspective},  where $X^\delta_0=\delta i_m$, $X^\delta_t \in \{\delta i_m, \ldots, \delta N\}$, together with the indicator $(\xi^\delta_t)_{t\ge 0}, \xi^\delta_t \in \{0,1\}$ of whether the particle is on or off transport, respectively.  Let
\[ H'_n := \inf\{t > 0: (X^\delta_t, \xi^\delta_t)=(\delta(i_*+1),1)\}. \]
Note that $H'_n=\inf\{t>0: X^\delta_t\in \delta R_n\}$, since particles cannot skip sections, and always enter a section on transport. If we were to analogously define a sequence of times $\{H_n^{i}\}_{i\in I^s_*}$ where for each $ i\in I^s_*$,  $H_n^{i}=\inf\{t>0: X^\delta_t\in \delta R_n\}$ with  $X^\delta_0=\delta i$ and $\xi^\delta_0\in\{0,1\}$,  then the exact first hitting time of the target region will satisfy $H_n=\min\{H_n^{i}: i\in I^s_*\}$. Hence, clearly $H_n\le H^{i_m}_n\equiv H'_n$. 

Since the exact distribution of $H_n^{i}$ is complicated, $\E{H_n}$ is analytically intractable, and instead we focus on finding a simple expression for $\E{H'_n}$. Our point is that, in the sparse material limit ($q_0$ small), the rightmost particle becomes increasingly likely to be the first successful particle to detach in the target region, and $H_n$ approaches $H'_n$.  

The computation of $\E{H'_n}$ requires computing the time it takes a particle to travel a certain axonal distance, given by the following.

\begin{lemma} \label{lem:travel-time}
Let the initial position of the particle be $(X^\delta_0, \xi^\delta_0)=(0,1)$ and let $L_* \in \{1,\dots,L\}$ be given where $L$ is the total length of the axon. Then, the time $T_{L_*} = \inf \{t > 0 : X^\delta_t =L_*\}$, satisfies
\[\E{T_{L_*}} = \frac{L_*(k_2+k_1)}{vk_1}.\]
\end{lemma}
\begin{proof}
Since $(X^\delta_t,\xi^\delta_t)_{t\ge 0}$ is a Markov process with generator 
\begin{align*}A^\delta f(x,\xi)&=r\xi\big[f(x+\delta,\xi)-f(x,\xi)\big]\\& + k_2\xi\big[f(x,\xi-1)-f(x,\xi)\big]+ k_1(1-\xi)\big[f(x,\xi+1)-f(x,\xi)\big]\end{align*}
it follows that
$M^1_t:=X^\delta_t-r\delta\int_0^t\xi^\delta_sds$ and $M^2_t:=\xi^\delta_t+\int_0^t k_2\xi^\delta_sds-\int_0^tk_1(1-\xi^\delta_s)ds$ are both martingales. Since both $M^1$ and $M^2$ have bounded increments and $\E{T_{L_*}}<\infty$, the optional stopping theorem implies that
\[0=\E{M^1_0}=\E{ M^1_{T_{L_*}} }=L_*-v\mathbb{E}\Big[\int_0^{T_{L_*}}\xi^\delta_sds\Big] \Rightarrow E\Big[\int_0^{T_{L_*}}\xi^\delta_sds\Big]=L_*/v\]
and 
\[1=\E{M^2_0}=\E{M^2_{T_{L_*}}}=1-k_1\E{T_{L_*}}+(k_2+k_1)\mathbb{E}\Big[\int_0^{T_{L_*}}\xi^\delta_sds \Big] \]
which implies 
\[ \E{T_{L_*}}=\frac{k_2+k_1}{k_1}\mathbb{E}\Big[\int_0^{T_{L_*}}\xi^\delta_sds\Big].\] 
and our claim follows.
\end{proof}

We also note that the same computation holds for the mean time a particle in the stochastic non-compartmental model $(X_t,\xi_t)_{t\ge 0}$ takes to reach a distance $L_*$, as the two martingales used in the proof depend only on $v$ and not on $\delta$.

We next compute the hitting time $H'_n$ of a particle started at location $i_m\in I_*$.

\begin{lemma} \label{lem:hitting-time}
Let the system $\{(Q_i(0), P_i(0)), i =1, \dots, N\}$ have the stationary distribution given by Proposition \ref{prop:equilib} conditional on $Q_i(0)=P_i(0)=0,$ for all $i \in R_n$. Then
\begin{align*}
	\E{H'_n} &= (1-e^{-\lambda_n})\left[\frac{1}{rq_0p_n} +\frac{k_2}{k_1(k_2+k_1)}\right] +  \frac{1}{k_2 p_n} \left[1-\left(\frac{r}{k_2 + r}\right)^n \left(1 + \frac{n k_2} {k_2 + r}\right)\right]\end{align*}
where $p_n = 1-(\frac{r}{k_2+r})^{n}$ and $\lambda_n= \frac{k_2+k_1}{k_1}q_0i_*p_n$.
\end{lemma}

\begin{proof} The proof of Proposition~\ref{prop:equilib} shows that conditioning on the values of $P_i(0), i\in R^n$ does not affect the law of $\{(Q_i(t), P_i(t)), i\le i_*\}$, hence for any $t\ge 0$, they form two mutually independent sequences of Pois($q_0$) and Pois($\frac{k_2 q_0}{k_1}$) random variables. 

Let $S$ denote the number of successful particles between sites $0$ and $i_*$ at time zero. The total of number of particles at time $0$ at these sites that are either on or off transport is distributed as a Pois$(\frac{k_2+k_1}{k_1} q_0i_*)$ variable. We next compute  the probability $p_n$ that any particle once it reached the target region is ``successful'' in detaching there. This probability can be written $p_n:=\sum_{i=1}^np(i)$ where $p(i)$ is the probability it first detaches at location $i_*+i$. Since  $p(i)=(\frac{r}{k_2+r})^{i-1}\frac{k_2}{k_2+r}$ is the chance a particle gets off transport in the $i$-th compartment, $p_n=\sum_{i=1}^np(i)=1-(\frac{r}{k_2+r})^{n}$.  
The probability that a particle is successful does not depend on which location it was at time $0$, and whether it was on or off transport at that time. 
Hence, $S$ is is distributed as a Pois$(\frac{k_2+k_1}{k_1} q_0i_*p_n)$ variable, and conditioned on the value $S$, the set of locations of the particles at time 0 is distributed as a set of $S$ draws from the uniform distribution on $\{1, \ldots, i_*\}$.  Since $i_m$ is the maximum of $S$ samples from a Uniform$\{1,\dots, i_*\}$ distribution we have
$ \E{i_m|S}=i_*-\frac{1}{i_*^S}\sum_{x=1}^{i_*-1}x^S \approx i_*\frac{S}{S+1}$. 

We now decompose $H'_n = H'_{e,n}+H'_{o,n}$, where we let $H'_{e,n}$ is the time it takes a successful particle initially at location $i_m$ to enter the region $R_n$, and $H_{o,n}$ is the time it takes any successful particle after it enters the region to get off transport. Thus, $\E{H'_n}=\E{H'_{e,n}}+\E{H'_{o,n}}$. 

If a successful particle starts at location $X^\delta_0=i_m\in\{0,\dots, i_*\}$ on transport $\xi^\delta_0=1$, then the time it takes to enter the region $R^n$ is by Lemma~\ref{lem:travel-time}
\[ 
\E{H'_{e,n}|(X^\delta_0,\xi^\delta_0)=(i_m,1)} = \frac{(i_*-i_m) \delta (k_2 + k_1)}{v k_1} 
\]
If it starts at location $X^\delta_0=i_m\in\{1,\dots, i_*\}$ off transport $\xi^\delta_0=1$ then it takes an additional Exponential time with mean $1/k_1$ for it to get back on transport at the same location, so $\E{H'_{e,n}|(X^\delta_0,\xi^\delta_0)=(i_m,0)} =\E{H'_{e,n}|(X^\delta_0,\xi^\delta_0)=(i_m,1)}+ \frac{1}{k_1}$, and since we assume new particles always enter the system on transport 
\[ 
\E{H'_{e,n}|i_m} = \frac{(i_*-i_m) \delta (k_2 + k_1)}{v k_1} + \frac{1}{k_1}\mathbf{1}_{i_m>0}.
\]
Since $\E{i_m|S} =  i_* \frac{S}{S+1}$, $\p{i_m>0}=\p{S>0}$, we get that
\[ 
\E{H'_{e,n}} = \E{\frac{1}{S+1}}\frac{i_* (k_2 + k_1)}{r k_1} +\frac{k_2}{k_1(k_2+k_1)}\p{S>0}
\]
with $S\sim\mbox{Pois}(\lambda_n)$, $\lambda_n= \frac{k_2+k_1}{k_1}q_0i_*p_n$. Since $\E{S}=\frac{1-e^{-\lambda_n}}{\lambda_n}$,  $\p{S>0}=1-e^{-\lambda_n}$, we get
\[ 
\E{H'_{e,n}} = (1-e^{-\lambda_n})\left(\frac{1}{rq_0p_n} +\frac{k_2}{k_1(k_2+k_1)}\right)
\]

To calculate $\E{H'_{o,n}}$ note that if a particle first gets off transport in the $i$-th compartment then the time of its travel until this point is a sum of $i$ independent and identically distributed exponential random variables with  parameter $k_2+r$. The probability a successful particle first gets off transport in the $i$-th compartment is $\frac{p(i)}{p_n}$. Hence, the time a successful particle takes to get off transport once it enters the region $R$ has the mean
\[\E{H'_{o,n}}=\sum\limits_{i=1}^n \frac{p(i)}{p_n} \frac{i}{k_2+r} = \frac{1}{k_2p_n}\left[1-\left(\frac{r}{k_2 + r}\right)^n\left(1 + \frac{n k_2}
{k_2 + r}\right)\right].\]
and our claim follows.
\end{proof}

To complete our analysis, we wish to characterize this result in terms of length along the axon and independent of the stepping size $\delta$. To this end, we fix a length $\ell$ and for a given compartment size $\ell$,  we let $n=\lceil \ell / \delta \rceil$, and for the start of the region we let $i_*=\lceil \ell_* / \delta \rceil$. As such, as $\delta\to 0$ the limiting region becomes $R_\ell=(\ell_*,\ell_*+\ell)$. Then, under the assumption that $q_0/\delta\to \rho_0$ we have 
\[p_n \to p_{\ell}= 1-e^{-\ell k_2/v},\quad \lambda_n\to \lambda_{\ell}=\frac{k_2+k_1}{k_1}\rho_0\ell_*p_\ell\]
and $\E{H'_n} \to \E{H'}$ where
\begin{equation} \label{eq:H-prime}
\E{H'}= \big(1-e^{-\lambda_\ell}\big)\left[\frac{1}{v \rho_0 p_\ell}+  \frac{k_2}{k_1(k_2+k_1)}\right]+\frac{1}{k_2 p_\ell}\left[1-e^{-\frac{\ell k_2}{v}}(1+\frac{\ell k_2}{v})\right].
\end{equation}

It remains to interpret this result with respect to the parameter choices we have made. First, we must set a value for the size $\ell$ of the target region. For this purpose we note that the typical size of a Node of Ranvier -- a gap in the myelin sheath of a myelinated axon where sodium channels are concentrated in the cell membrane -- is approximately one micron.\footnote[1]{We do not wish to claim this is a complete model for deposition of sodium channels in Nodes of Ranvier, as the particulars of the biology -- which may include factors like local signals that encourage motors to detach from microtubules near the Nodes -- are not fully understood.  We merely wish to use the size of the Nodes to fix our intuition about the size of a target to other important length scales such as that of microtubules which are also a micron in length.}

As seen in the preceding proof there are three contributions to $\E{H'}$, and we next analyze them in terms of their dominance for the overall value. We begin with the contribution from the last term, which is the time it takes for a successful particle to detach once it has reached the target region. By our choice of $\ell$, the recurring ratio $\ell / v$ is one. Along with the earlier assumption that $k_2 = 1$, the entire term simplifies to $(e-2)/(e-1) = 0.4 \, s$. The multiplicative factor $1 - e^{-\lambda_\ell}$ preceding the first term of \eqref{eq:H-prime} results from a boundary effect: when $\ell_*$ is very close to zero, it is very unlikely there are any particles already in the system between the soma and the target region. When the target region is in the middle of the axon, this contribution from particles in that section of the axon at time 0 is significant.

When $q_0 \sim 1$ as assumed earlier, then $\rho_0=q_0/\delta\sim 10^8$. Then if $\ell_*> 10^{-8}$, which corresponds to the target region being just one motor step down the length of the axon, we have $\rho_0\ell_*\sim 1$ and $\lambda_\ell>2.3$ so $1-e^{-\lambda_\ell}>0.9$. Looking at the first term inside the parenthesis, we note that $v\rho_0\sim 10^2$, while $p_\ell = 1 - \exp(-\ell k_2 / v) \approx 0.6$ is the probability that any given vesicle will be successful in detaching from transport in the target region, so $1/(v\rho_0 p_\ell) \sim 10^{-2}$. Meanwhile, the term $k_2 / (k_1 (k_2 + k_1))$, which is the expected time to bind to transport if a successful particle happens to be off transport at time 0, is 1/4 for the assumed values of $k_2$ and $k_1$. Therefore under the $q_0 \sim 1$ assumption, both this and the final term contribute significantly to the hitting time.

However, in the sparse material regime, say $q_0 \sim 10^{-3}$, we have then $\rho_0 \sim 10^5$, and the factor $1 - e^{-\lambda_\ell}>0.9$ when $\ell_* > 10^{-5}$. That is, if the target region is at least 1/100, or at least 1000 segments, down the length of a $10^{-3} m$ axon.  However, now the rate limiting factor is the wait time for the first successful vesicle to arrive in the target region, which is captured by the term $1/v \rho_0 p$. The product $v \rho_0 \approx 0.1$ measures the average rate at which new particles should arrive and together with the probability of success $p_\ell\approx 0.6$, we now have $1 / v \rho_0 p_\ell \approx 16 s$. So, in the sparse material regime the first summand in $\E{H'}$ giving the mean time for arrival of the particles to the target region dominates.

It is interesting to consider what happens to the arrival rate term under perturbations of the various parameters.  In particular, we draw the reader's attention to changes in $v$.  When viewing intermittent search in terms of a single particle, the probability of finding the target is strictly decreasing in $v$. Higher velocity seems to be the enemy of finding the target.  Indeed, from a system point of view, this implies that the system will require more trials before a successful particle arrives in the target region. What the equation \eqref{eq:H-prime} gives us is the ability assess how much more quickly the trials will happen.  In fact, the function $v (1 - \exp(-\ell k_2 /v)) $ is increasing in $v$. Therefore the entire expected wait time $\E{H'}$ is actually \emph{decreasing} in $v$.  That is to say, while the particles are less likely to succeed, they will be arriving more rapidly enough to counterbalance the lost time.  We believe that this kind of quantitative analysis will prove fruitful in future study when coupled with more details of the biology of deposition of materials in the cell membrane.

\subsection{Approaching Equilibrium}
\label{sec:approach-to-equilib}

We have seen above that the axon is very homogeneous at stochastic equilibrium on a space scale down to micrometers.  One of the beautiful properties of transport with reversible binding is that if it is locally out of equilibrium, the on-off dynamics can return the system to equilibrium on a much faster time scale than waiting for new material to arrive from the nucleus. Furthermore, as is discussed in \cite{welte08}, one of the key goals of biophysics investigation is the discovery of behaviors for which biochemical regulation is not necessary. This is of fundamental importance to the biological function of the system
because it means that the axon will automatically ``repair" itself
without central control of the repair process. 

How good is this mechanism? If a segment of the axon is far away from equilibrium, how long does it take to get back close to equilibrium?  
To investigate this question, we imagine that the axon is at stochastic equilibrium except for some segment $R$, where the total number of particles on and off transport, $Q_i(0)+ P_i(0)=0, \forall i\in R$, is zero at time $0$. Let $a$ and $\ep$ be given small numbers and suppose that $\lambda_{\infty}$ is the mean vector for $(Q_i, P_i)$ at stochasttic equilibrium.  We want to compute (an upper bound for) the time $t^*$ so that
\[         \mathbb{P}\{|(Q_i(t^*),P_i(t^*))- \lambda_{\infty}| \geq a|\lambda_{\infty}|, \forall i\in R\} \leq \ep,  \]
that is, the probability that $(Q_i,P_i), \forall i\in R$ is significantly different from $\lambda_{\infty}$ is very small. In the applications below we will choose $a = 0.1$ and $\ep = 0.05,$ and we will see that a 10 micron segment can recover in about 10 seconds, while a 1 millimeter segment will take about 1000 seconds or 15 minutes to recover. We study this question first for a single location $R=\{i\}$, and then use the estimates
derived to scale the results to segments of any length.

\begin{proposition} \label{prop:approach-equilib}
Let $(Q_i(0),P_i(0)) = (0,0)$ and let the constants $a > 0$ and $\ep \in
(0,1)$ satisfy the relationship $a^2 \ep |\lambdabf_\infty| > 1$ where
$\lambdabf_\infty=q_0 (1,\frac{k_2}{k_1})$  is the equilibrium vector of $(Q_i,P_i)$. Then
there exists $t_\ast>0$ such that $\forall t \geq t_*$,
$$
\p{|(Q_i(t),P_i(t)) - \lambdabf_\infty| \geq a |\lambdabf_\infty|} \leq \ep.
$$
In fact, the choice $
t_\ast = \alpha^{-1} \ln\Big(\frac{\sqrt{p |\lambdabf_\infty|}}{a
\sqrt{\ep |\lambdabf_\infty|} - 1}\Big) $
is sufficient, where 
\[\alpha := \frac{1}{2} \left( k_2 + k_1 + r - \sqrt{(k_2 + k_1 + r)^2
- 4 k_1 r}\right).\]
\end{proposition}

We note that given a particular choice of $a$ and $\ep$, the  condition $a^2\ep|\lambda_{\infty}| > 1$ guarantees that there are ``enough'' particles.

\begin{proof} Note that for any given $\beta \in (0,1)$, we may
choose $t_\ast > 0$ such that for all $t \geq t_\ast$, the vector of
means $\lambdabf(t):=\E{(Q_i(t),P_i(t))}$ satisfies $|\lambdabf(t) -
\lambdabf_\infty| \leq a \beta |\lambdabf_\infty|$. Then
\begin{align*}
\p{|(Q_i(t),P_i(t)) - \lambdabf_\infty| \geq a |\lambdabf_\infty|} &\leq
\p{|(Q_i(t),P_i(t)) - \lambdabf(t)| + |\lambdabf(t)
-\lambdabf_\infty| \geq a |\lambdabf_\infty|} \\
&\leq \p{|(Q_i(t),P_i(t)) - \lambdabf(t)| \geq a (1 - \beta)
|\lambdabf_\infty|}
\end{align*}
Applying Chebyshev's Inequality, and observing that the variance of
a Poisson random variable is equal to its mean, we conclude that
$$
\p{|(Q_i(t),P_i(t)) - \lambdabf_\infty| \geq a |\lambdabf_\infty|} \leq
\frac{\V{|(Q_i(t),P_i(t))|}}{a^2 (1 - \beta)^2 |\lambdabf_\infty|^2} =
\frac{|\lambdabf(t)|}{a^2 (1 - \beta)^2|\lambdabf_\infty|^2}
$$
Since the initial condition for both $P_i$ and $Q_i$ are less than their
respective equilibrium values, each are monotonically increasing
functions and the above reduces to
$$
\p{|(Q_i(t),P_i(t)) - \lambdabf_\infty| \geq a |\lambdabf_\infty|} \leq
\frac{1}{a^2 (1 - \beta)^2 |\lambdabf_\infty|}
$$
for all $t > t_\ast$. In order to satisfy the requirement that the
right hand side must be less than $\ep$, we solve for $\beta$ and find $
\beta = 1 - (a \sqrt{\ep |\lambda_\infty|})^{-1} $
provided that $a \sqrt{\ep |\lambda_\infty|} > 1$.

It remains to study the convergence of the mean and the appropriate
choice of $t_\ast$. The dynamics of the mean vector $\lambdabf(t)$
are given by the ODE
\begin{equation} \label{eq:lambda-2d-ode}
\frac{d}{dt} \lambdabf (t) = - A_1 \lambdabf (t) + q_0 r {\ebf}_1
\end{equation}
where ${\ebf}_1$ is the unit vector $(1,0)$ and $
A_1 = \left(\begin{array}{cc} k_2 + r & - k_1 \\ - k_2 & k_1
\end{array}\right).$

The solution to \eqref{eq:lambda-2d-ode} is $
{\lambdabf} (t) = {\lambdabf}_\infty + e^{- A_1 t} (\lambdabf (0) -
{\lambdabf}_\infty) $
where ${\lambdabf}_\infty = q_0 r A_1^{-1} \ebf_1 = q_0 (1,
\frac{k_2}{k_1})$. This yields the estimate
$$
|\lambdabf (t) - {\lambdabf}_\infty| \leq \left| e^{-A_1 t}
{\lambdabf}_\infty \right| \leq e^{-\alpha t} |{\lambdabf}_\infty|
$$
where $\alpha$ is the smaller of the eigenvalues of $A_1$. 
Noting that $\alpha > 0$, $t_\ast$ may be chosen so that $e^{-\alpha
t_\ast} \leq a \beta$, i.e. $ t_\ast = \alpha^{-1} \ln (1/(a \beta))$
is an upper bound for the time to be close to equilibrium with high probability.
\end{proof}

In order to calculate return to equilibrium at various scales, we now suppose that the whole axon is in statistical equilibrium except for a segment $R$ of length $\Delta=\delta10^{\nu}$ in which we will assume that there are no particles either on or off the track. Proposition \ref{prop:approach-equilib} covered the case $\nu =0$. We are interested in $\nu = 1, \ldots, 8.$ We imagine that the axon is broken up into $10^{8-\nu}$ segments of length $\Delta$. In this rescaled system, the unbinding and binding rates per particle, $k_2$ and $k_1$ remain the same, as well as the mean on-transport velocity $v$. In order to retain this mean
velocity, the rate of lateral stepping must be decreased to $\tilde r = r 10^{-\nu}$.

The ODE for the mean vector of the rescaled system is given by
\begin{equation} \label{eq:lambda-2d-ode-rescaled}
\frac{d}{dt} \tilde \lambdabf(t) = - \tilde A_1 \tilde \lambdabf(t)
+ q_0 r \ebf_1
\end{equation}
with $ \tilde A_1 = \left(\begin{array}{cc} k_2 + \tilde r & - k_1 \\
- k_2 & k_1 \end{array}\right). $

We note that the last term in \eqref{eq:lambda-2d-ode-rescaled} contains an $r$ rather than an $\tilde r$.  This is because the input rate is unchanged while the exit rate is diminished. The resulting equilibrium value is therefore rescaled as well, $
\tilde \lambdabf_\infty = q_0 r \tilde A_1^{-1} \ebf_1 = q_0
\frac{r}{\tilde r} \binom{1}{k_2 / k_1}=q_0 10^{\nu} \binom{1}{k_2 / k_1}.$
Both components of this vector are of order $10^\nu$, as expected.

Using the parameters discussed in Section \ref{subsec:model}: $k_2 = 1, k_1 = \frac{1}{3}, v = 10^{-6}m/s, r = 10^2 s^{-1}$. This implies $\tilde r = 10^{2-\nu}s^{-1}.$ We choose the thresholds to be $a = 0.1$ and $\ep = 0.05$, so the constraint that $a^2 \ep |\lambdabf_\infty| > 1$ requires that $|\lambdabf_\infty| > 2 \times 10^3$. Since $0.25 \leq q_0 \leq 2.5$, this is indeed the case if $\nu > 3$, i.e. if the segment has length greater than 10 microns.

It follows from Proposition \ref{prop:approach-equilib} that the time to equilibrium, $\tilde t_\ast$, is proportional to $\tilde \alpha^{-1}$ where $\tilde \alpha$ satisfies:
\begin{align*}
\tilde \alpha &= \frac{1}{2} \left( k_2 + k_1 + \tilde r -
\sqrt{(k_2 + k_1 + \tilde r)^2 - 4 k_1 \tilde r}\right) \\
&= \frac{2 k_1 \tilde r}{k_2 + k_1 + \tilde r + \sqrt{(k_2 + k_1 +
\tilde r)^2 - 4 k_1 \tilde r}}.
\end{align*}
For the given parameter values (with $|\lambdabf_\infty| > 5 \times 10^3$ in particular), the constant of proportionality $\ln\left(\sqrt{\ep |\lambdabf_\infty|} / (a \sqrt{\ep |\lambdabf_\infty|}-1)\right)$ is contained in the interval $(2.3, 2.5)$ and does not have a impact on how the relaxation time scales with $\nu$. 
In terms of analyzing $\tilde \alpha$, we note that $k_1 \tilde r$ is small compared to $k_2$. To leading order, $ \tilde \alpha \sim (k_1 \tilde r)/(k_2 + k_1 + \tilde r) \sim
10^{2 - \nu} s^{-1} $. It follows that $\tilde t_\ast \sim 10^{\nu - 2} s.$
Thus, for a 10 micron segment ($\nu = 3$) the time to recover is about 10 seconds and for a 1 millimeter segment ($\nu = 5$) the time to recover is 1000 seconds or 15 minutes. The time to recover depends, of course, on the parameter $\ep$ that represents what we mean by ``close.'' We also note that one can compute various measures of time to recover using the PDE models discussed in Section 3. \\

\section{Discussion}

In this paper, we created a spatial Markov chain model for studying various aspects of fast axonal transport.  Previous models that use PDEs treat the velocity of transport as constant when particles are attached to the fast transport system. Since it is known that transport along the microtubules is itself stochastic, it is important to have a fully stochastic model. Our model allows us to unify and extend previous work. In Section 3.2 we show that from the particle perspective as the compartment size $\delta$ tends to zero, our model converges in distribution on the Skorokhod space of \emph{c\`adl\`ag} functions to the piecewise-deterministic model analyzed by Brooks \cite{B99}. Namely, we show that the paths of particles in our model converge to those of particles in a stochastic non-compartmental model. The argument proceeds by an explicit computation of the finite-dimensional distributions and a tightness argument. In Proposition \ref{prop:wavefront} we give a rigorous probabilistic proof of why the paths of particles  follow ``approximate traveling waves'' described by other authors \cite{reed90,B99,friedman05}.  This proof is based on stochastic averaging arguments which show that a functional central limit theorem holds on the space of continuous functions for the paths of particles as the compartment size decreases. The diffusion of particles around their mean position can consequently be approximately described jointly for all time by a Brownian motion with the appropriate diffusion coefficient.

In Section 2, we show how to use existing experimental data to indentify (ranges for) all the parameters of our model. In light of this, we can use the model to investigate several important biological questions. These are based on describing the spatial distribution of multiple particles in our model. In Section 4.1 we derive the stationary distribution for the number of particles in different compartments on and off transport along the axon. This gives an explicit description of the stochasticity of the system that is present even after a long time. In Section 4.2 we derive estimates for how homogeneous the axon is on different spatial scales. In Section 4.3 we study a question introduced by Bressloff, by providing a stochastic quantity which describes the balance the system needs to achieve between rapid transport that brings new material quickly and efficient local search that improves time of delivery to a target. 
Finally, in Section 4.4, we use the model to calculate the length of time that it would take for axonal segments of different lengths to recover to near stochastic equilibrium after they have been depleted of vesicles.

In our stochastic compartmental model all event wait
times are assumed to be exponential random variables, but this is certainly a
simplification. As an example, the stepping process of kinesin is a
well-studied though still not completely understood phenomenon. Much
work has focused on assessing the dependence of the mean rate of
translocation on both the load and the local concentration of ATP
\cite{block99} \cite{block00}. Implicit in this analysis is the
assumption of exponential wait times with state dependent rate
parameters. However, when fitting to data and matching dispersion
information the authors in \cite{fisher01} found it necessary to
generalize the wait time distribution.  This was followed by
more detailed models for which it was shown that load carrying could in fact regularize
the stepping times of kinesin motors \cite{schilstra06}
\cite{deville08}.
 Generalizing waiting times would significantly affect our results. Since the particle position process is no longer Markov, we no longer have the direct connection to the previous results stated in Section 3.1.-3.3., nor can we use the stationary distribution employed in Section 4.1 and used for addressing the biologicals questions in Sections 4.2.-4.4.
In light of the known need for generalized wait times in the
stepping process, it seems likely that detailed observation of the
rebinding process will call for new mathematical models as well. Recall that when
a vesicle unbinds from a microtubule it is unclear whether it
typically rebinds to the same microtubule or if it explores the
region significantly via diffusion before finding a different
microtubule to bind to.  In the latter case, a more appropriate
model for rebinding time would be to solve some kind of first
passage time problem and use that distribution for the rebinding
wait.

An important aspect of the biology of axonal transport is not included in the model presented here, namely the
local deposition and eventual degradation of transported materials. For example, sodium channels and sodium pumps are synthesized at the soma, transported down the axon and deposited in the axonal membrane, ether uniformly as in an unmyelinated axon or at the nodes of Ranvier in a myelinated axon.   Channels and pumps are proteins with half-lives on the order of days to weeks. The present model can be extended to include a deposition compartment at each location, and, clearly, the processes of deposition and subsequent degradation will cause the mean number of particles both on and off transport to be monotone decreasing as one moves down the axon. How inhomogeneous this makes the axon will depend on the details of deposition and degradation rates. Our preliminary calculations indicate that long axons, such as the meter-long axons in human sciatic nerve, would be quite inhomogeneous. This is an important biological issue because it is controversial whether the machinery for protein synthesis (i.e. ribosomes) exist in axons \cite{twiss09}. We have also not included retrograde transport or the fact that some axons may have location-dependent unbinding rates \cite{dixit08}. All of these issues will be the subject of future work.

\section*{Acknowledgments}
The authors are grateful to Professors H.~Frederik Nijhout and Vann
Bennett of Duke University and Professor Anthony Brown of The Ohio
State University for helpful discussions. This work was supported by NSF grants DMS-061670 and EF-1038593, and an NSERC Discovery Grant.


\begin{thebibliography}{10}

\bibitem{alberts08}
Alberts B, Johnson A, Lewis J, Raff M, Roberts K, Walter P (2008)
\newblock {\em The Molecular Biology of the Cell,}
\newblock Garland Science, Taylor \& Francis Group, New York.

\bibitem{allen85}
Allen RD, Weiss DG, Hayden JH, Brown DT, Fijiwaki H, Simpson M (1985).
\newblock Gliding movement and bidirectional transport along single native microtubules from squid axoplasm: evidence for an active role for microtubules in cytoplasmic transport.
\newblock {\em J. Cell Biol} {\bf 100}, 1736-1752.

\bibitem{blum85}
Blum JJ, Reed MC (1985)
\newblock A Model for Fast Axonal Transport.
\newblock {\em Cell Motility} {\bf 5}, 507-527.

\bibitem{bressloff06}
Bressloff PC, (2006)
\newblock Stochastic model of protein receptor trafficking prior to synaptogenesis
\newblock {\em Physical Review E} {\bf 74} 031910.

\bibitem{B99}
Brooks EA (1999)
\newblock Probabilistic methods for a linear reaction-hyperbolic system with constant coefficients.
 \newblock {\em Ann. Appl. Prob.}, {\bf 9(3)}, 719-731.

 \bibitem{brown00}
 Brown A (2000)
 \newblock Slow axonal transport: stop and go traffic in the axon.
 \newblock {\em Nat Rev Cell Mol Biol} {\bf 1}, 153-156.

 \bibitem{carter05}
 Carter NJ, Cross RA (2005)
 \newblock Mechanics of the kinesin step.
 \newblock {\em Nature} {\bf 435}, 308-312.

 \bibitem{cox62}
 Cox DR (1962)
 \newblock {\em Renewal Theory.}
 \newblock Methuen \& C0. Ltd., London, Section 4.5.

\bibitem{deville08}
DeVille RL and Vanden-Eijnden E (2008)
\newblock Regular Gaits and Optimal Velocities for Motor Proteins
\newblock {\em Biophysics Journal} {\bf 95} 6: 2681-2691.

  \bibitem{dixit08}
 Dixit R, Ross JL, Goldman YE, Holzbaur ELF (2008)
 \newblock Differential regulation of dynein and kinesin otor proteins by tau.
 \newblock {\em Science} {\bf 319}, 1086-1089.

  \bibitem{ek86}
 Ethier S, Kurtz TG (1986)
 \newblock {\em Markov Processes: Characterization and Convergence.}
 \newblock John Wiley \& Sons Inc., New Jersey


 \bibitem{finer94}
 Finer JT, Simmons RM, Spudich JA (1994)
 \newblock Single myosin molecule mechanics: piconewton forces and nanometre steps.
 \newblock {\em Nature} {\bf 368}, 113-119.

\bibitem{fisher01}
 Fisher ME, Kolomeisky (2001)
 \newblock Simple mechanochemistry describes the dynamics of kinesin
 molecules.
 \newblock {\em PNAS} {\bf 98} 14: 7748-7753.

  \bibitem{friedman05}
 Friedman A, Craciun G (2005)
 \newblock A model of intracellular transport of particles in an axon.
 \newblock {\em SIAM J Appl Math} {\bf 38}, 741-758.

 \bibitem{friedman06}
 Friedman A, Craciun G (2006)
 \newblock Approximate traveling waves in linear reaction-hyperbolic equations.
 \newblock {\em SIAM J Appl Math} {\bf 38}, 741-758.


  \bibitem{friedman07}
 Friedman A, Hu B (2007)
 \newblock Uniform convergence for approximate traveling waves in reaction-hyperbolic systems.
 \newblock {\em Indiana Math J} {\bf 56}, 2133-2158.

   \bibitem{friedman07-2}
 Friedman A, Hu B (2007)
 \newblock Uniform convergence for approximate traveling waves in reaction-diffusion-hyperbolic systems.
 \newblock {\em Arch Rat Mech Anal} {\bf 186}, 251-274.

 \bibitem{gennerich07}
     Gennerich A, Carter AP, Reck-Peterson, Vale RD (2007)
 \newblock Force-induced bidirectional stepping of cytoplasmic dynein.
 \newblock {\em Cell} {\bf 131}, 952-965.

\bibitem{GH69}
	Griego RJ, Hersh R (1969)
\newblock ``Random evolutions, Markov chains, and systems of partial differential equations.''
\newblock {\em Proc. Nat. Acad. Sci. USA} {\bf 62}, 305-308.

\bibitem{gross82}
Gross GW, Weiss, DG (1982)
\newblock ``Theoretical considerations on rapid transport in low viscosity axonal regions.''
\newblock \emph{Axoplasmic Transport} (Ed. D.G. Weiss),
Springer-Verlag, Berlin.

\bibitem{jastrow-url}
Jastrow H
\newblock Electron Microscopic Atlas of Cells, Tissues and Organs on
the Internet.  Schwann cells section. \\
\newblock \texttt{http://www.uni-mainz.de/FB/Medizin/Anatomie/workshop/EM/EMSchwannE.html}

\bibitem{Kal}
Kallenberg O (2002)
\newblock ``Foundations of modern probability", 2nd ed.
\newblock Springer-Verlag, New York.

\bibitem{Kh66a}
Khas'minskii, RZ (1966)
\newblock ``On stochastic processes defined by differential equations with a small parameter.''
\newblock \emph{Theory Probab. Appl. }, {\bf 11}, 211-228.

\bibitem{Kh66b}
Khas'minskii RZ (1966)
\newblock ``A limit theorem for the solutions of differential equations with random right-hand sides.''
\newblock \emph{Theory Probab. Appl. }, {\bf 11}, 390-406.

\bibitem{K79}
Kelly FR
\newblock {\em Reversibility and stochastic networks}.
\newblock John Wiley \& Sons Ltd., Chichester, 1979.
\newblock Wiley Series in Probability and Mathematical Statistics.

\bibitem{k92}
Kurtz TG (1992)
\newblock Averaging for martingale problems and stochastic 
approximation.
 \newblock {\em Lecture Notes in Control and Inform. Sci. }, {\bf 177}, 186-209.


\bibitem{lasek82}
Lasek RJ, Brady ST (1982).
\newblock The structural hypothesis of axonal transport: two classes of moving elements.
\newblock {\em Axoplasmic Transport} (ed. G. Weiss), Springer-Verlag, Berlin, 1982, pp. 397-405.

\bibitem{lawler95}
Lawler G (1995)
\newblock {\em Introduction to Stochastic Processes}
\newblock{\em } Chapman and Hall, New York.


\bibitem{miller85}
Miller R, Lasek RJ (1985)
\newblock Crossbridges mediate anterograde and retrograde vesicle transport along microtubules in squid axoplasm.
\newblock {\em  J Cell Biol} {\bf 101}, 2181-2193.

\bibitem{moranurl}
Moran D, Rowley JC III
\newblock Visual Histology.com, Chapter 8, Nerves. \\
\newblock \texttt{http://www.visualhistology.com/products/atlas/VHA\_Chpt8\_Nerves.html}

\bibitem{newby09}
Newby JM, Bressloff PC
\newblock Directed Intermittent Search for Hidden Targets 
\newblock {\em New Journal of Physics} {\bf 11} 023033.

\bibitem{newby10}
Newby JM, Bressloff PC
\newblock Quasi-steady State Reduction of Molecular Motor-Based Models of Directed Intermittant Search
\newblock {\em Bullentin of Mathematical Biology} {\bf 72} 1840-1866.

\bibitem{ochs72}
Ochs S (1972)
\newblock Rate of fast axoplasmic transport in mammalian nerve fibers.
\newblock {\em J Physiol} {\bf 227}, 627-245.

\bibitem{reed86}
Reed MC,  Blum JJ (1986)
\newblock Theoretical analysis of radioactivity proÞles during fast
axonal transport: effects of deposition and turnover.
\newblock{\em Cell Motility and the Cytoskeleton}, {\bf 6}, 620Ð627.

\bibitem{reed94}
Reed MC, Blum J (1994)
\newblock Mathematical Questions in Axonal Transport.
\newblock In {\em Lectures in mathematics in the Life Sciences,}{\bf 24},
\newblock Amer. Math. Soc., Providence.

\bibitem{reed90}
Reed MC, Venakides S, Blum JJ (1990)
\newblock Approximate traveling waves in linear reaction-hyperbolic equations.
\newblock {\em SIAM J. Appl. Math.}, {\bf 50}, 167Ð180.

\bibitem{schilstra06}
Schilstra MJ, Martin SR (2006)
\newblock An elastically tethered viscous load imposes a regular
gait on the motion of myosin-V.  Simulation of the effect of
transient force relaxation on a stochastic process.
\newblock {\em J.R.~Soc.~Interface} {\bf 3} 153-165.

\bibitem{block00}
Schnitzer MJ, Visscher K, Block SM (2000)
\newblock Force production by single kinesin motors.
\newblock {\em Nat. Cell. Biol.} {\bf 2} 718-722.

\bibitem{stewart82}
Stewart GH, Horwitz B, Gross GW (1982)
\newblock A chromatographic model of axoplasmic transport. In
\newblock {\em Axoplasmic Transport} (ed. G. Weiss), Springer-Verlag, Berlin, 1982, 414-422.

\bibitem{svoboda93}
Svoboda K, Schmidt C, Schnapp B  (1985)
\newblock Direct observation of kinesin stepping by optical trapping
interferometry.
\newblock {\em  Nature} {\bf 365}, 721-727.

\bibitem{takenaka84}
Takenaka T, Gotoh H (1984)
\newblock Simulation of axoplasmic transport.
\newblock {\em J Theor Biol} {\bf 107}, 579-601.

\bibitem{twiss09}
Twiss JL, Fainzilber M (2009)
\newblock Ribosomes in axons - scrounging from the neighbors.
\newblock {\em Trends in cell biology.} {\bf 19}, 236-243

\bibitem{vale85}
Vale RD, Reese TS, Sheetz MP (1985).
\newblock Identification of a novel force-generating protein, kinesin, involved in
microtubule-based motility.
\newblock {\em Cell} {\bf 42}, 39-50.

\bibitem{block99}
Visscher K, Schnitzer MJ, Block SM (1999)
\newblock Single kinesin molecules studied with a molecular force
clamp
\newblock {\em Nature} {\bf 400}, 184-189.

\bibitem{welte08}
Welte MA, Gross SP (2008)
\newblock Molecular motors: a traffic cop within?
\newblock {\em HFSP Journal} {\bf 2} 4: 178-182. 

\bibitem{willis81}
Willis WD, Grossman RG (1981)
\newblock \emph{Medical Neurobiology: Neuroanatomical and Neurophysiological Principles
Basic to Clinical Neuroscience}.
\newblock Mosby, Inc.

\end{thebibliography}
\end{document}